\RequirePackage{fix-cm}
\documentclass[smallextended, natbib]{svjour3}       
\pdfoutput=1
\smartqed  
\usepackage{graphicx}
\usepackage{mathptmx}      
\usepackage{amsfonts}
\usepackage{amsmath, bm}
\def\D{\mathrm{d}}
\journalname{Mathematical Biology}
\begin{document}

\title{Robustness of movement models: can models bridge the gap between temporal scales of data sets and behavioural processes?
}

\titlerunning{Robustness of discrete-time movement models}        

\author{Ulrike E.\ Schl\"agel         \and
        Mark A.\ Lewis 
}


\institute{U.E.\ Schl\"agel \at
              Department of Mathematical and Statistical Sciences, CAB 632, University of Alberta, Edmonton, AB T6G 2G1, Canada \\
              \email{ulrike.schlaegel@gmail.com}             
           \and
           M.A.\ Lewis \at
              Department of Mathematical and Statistical Sciences, CAB 632, University of Alberta, Edmonton, AB T6G 2G1, Canada
}

\date{Received: date / Accepted: date}

\maketitle

\begin{abstract}
Discrete-time random walks and their extensions are common tools for analyzing animal movement data. In these analyses, resolution of temporal discretization is a critical feature. Ideally, a model both mirrors the relevant temporal scale of the biological process of interest and matches the data sampling rate. Challenges arise when resolution of data is too coarse due to technological constraints, or when we wish to extrapolate results or compare results obtained from data with different resolutions. Drawing loosely on the concept of robustness in statistics, we propose a rigorous mathematical framework for studying movement models' robustness against changes in temporal resolution. In this framework, we define varying levels of robustness as formal model properties, focusing on random walk models with spatially-explicit component. With the new framework, we can investigate whether models can validly be applied to data across varying temporal resolutions and how we can account for these different resolutions in statistical inference results. We apply the new framework to movement-based resource selection models, demonstrating both analytical and numerical calculations, as well as a Monte Carlo simulation approach. While exact robustness is rare, the concept of approximate robustness provides a promising new direction for analyzing movement models.
\end{abstract}

\section{Introduction}
\label{s;intro}

Major advances in tracking technology during the last decades have made large datasets of animal movement available to ecologists, and analyses of data have become widespread in ecology. These analyses have shed light on mechanisms that underly fundamental processes such as migration \citep{Robinson:2009br,Costa:2012ba}, navigation \citep{Tsoar:2011wz,Benhamou:2011jw}, or home range behaviour and territoriality \citep{Borger:2008gc,Potts:2014dr,Giuggioli:2014tc}. They  have helped to identify conservation goals by revealing habitat preferences and critical environmental features for populations \citep{Sawyer:2009uq,Colchero:2010cb,Ito:2013hg,Masden:2012bp}, as well as the role of important mutualistic interactions between mobile animals and immobile plants \citep{Cortes:2013um,Mueller:2014hq}. These are only few of the many facets of movement ecology. 

Mathematical and statistical models provide a framework for studying movement \citep{Schick:2008dn,Smouse:2010ku,Langrock:2012ek}. When linking models to data, we can estimate model parameters and identify best-fitting models, thus inferring unknown quantities or mechanisms in movement behaviour. 
Although movement itself is a continuous process, many individual-based movement models treat time as a discrete variable, viewing movement as a series of locations in space, or equivalently as a series of steps \citep{Turchin:1998td,McClintock:2014in}. This may largely be ascribed to data being available in this format. Discrete-time models may thus be an intuitive first choice to describe a sampled movement path. However, there may be more reasons to use discrete-time models. The continuous movement path of an animal may consist of various behaviours at different scales \citep{Johnson:2002hs,Benhamou:2013de}. Using a discrete-time model at the scale of interest allows us to focus on the behavioural mechanisms at that scale, while, for example, combining other unknown processes as stochastic effects. Also, the choice of time formulation in a movement model can have side effects that impact inference results. For example, \cite{McClintock:2014in} demonstrated that using a continuous-time Ornstein-Uhlenbeck process in a hierarchical model for identifying behavioural states led to difficulties discriminating between states, due to an inherent correlation between the variables step length and bearing in the Ornstein-Uhlenbeck process.

Many movement models are based on random walks \citep{Turchin:1998td,Codling:2008js,McClintock:2014in}. From simple random walks that assume independently and identically distributed steps, we have moved to correlated random walks, which include directional persistence \citep{Kareiva:1983vf}, and biased random walks, which can, for example, be used to model centralizing tendencies or long-term directional goals \citep{Benhamou:2006kq,Borger:2008gc,McClintock:2012bw}. Many animals live in heterogeneous environments, and the composition of the environment and availability of resources influence movement decisions \citep{Fortin:2005uk,McPhee:2012ii}. Therefore, another trend of random-walk extensions has left behind  assumptions about environmental homogeneity in favour of spatially-explicit models that incorporate habitat effects on movement decisions \citep{Rhodes:2005uu,Avgar:2013hz,Potts:2014wt}. These models have an advantage over statistical resource-selection and step-selection functions \citep{Manly:2002un,Fortin:2005uk,Forester:2009wa} by allowing simultaneous estimation of movement parameters and environmental effects.

When linking discrete-time models to data, the temporal resolution of the discretization is a critical feature that must be chosen with care. Different time scales may come into play and need to be consolidated. On the one hand, a time scale is given by the biological process of interest. For example, we may be interested in inferring behavioural mechanisms of a movement process and thus need to consider the time scale at which these mechanisms are relevant. The discretization of a model should represent this scale appropriately. On the other hand, a different time scale may be given by the data collection rate. In practice, the sampling rate of data is subject to technological constraints. One of the major limitations of electronic tagging devices such as Argos or GPS tags is battery life, imposing a tradeoff between measurement rate and total deployment time \citep{Ryan:2004et,Breed:2011ux}. Also, to avoid a large noise to signal ratio, the time interval should be chosen so that measurement error relative to distance travelled during a time interval is small \citep{Ryan:2004et}. For slow moving animals and depending on the accuracy of the tagging device, a minimum time interval of an hour may be necessary \citep{Jerde:2005fp}. Therefore, the resolution of the data may not always match the time scale of the behavioural process of interest. In this case, it becomes a challenge for a model to overcome the conflict. 

A related problem is that sampling rate can affect data analysis results \citep{Codling:2005uo,Rowcliffe:2012hp,Postlethwaite:2013ki}. A common measure calculated from raw movement data is the total distance travelled, which can provide useful information about an animal's energetic expenditures. It is well documented that this quantity is highly influenced by the sampling rate of the data \citep{Ryan:2004et,Mills:2006tn,Tanferna:2012fc,Rowcliffe:2012hp}. A range of studies demonstrated that other fundamental movement characteristics vary with data sampling frequency as well, for example path sinuosity and apparent speed  \citep{Codling:2005uo}, movement rate and turning angle \citep{Postlethwaite:2013ki}, and estimates of territory size \citep{Mills:2006tn}. One of the main problems underlying these effects is information loss when subsampling a movement path. This also impairs our capacity to correctly estimate behavioural states through hierarchical modelling approaches that have become widespread in movement analyses \citep{Breed:2011ux,Rowcliffe:2012hp}. These findings demonstrate that great care is needed when extrapolating movement analysis results beyond the temporal scale of a study. Comparisons of results may not be appropriate if the temporal resolution of the data varies too much, but it is unclear what constitutes `too much'.

Despite the evidence of the extent of the problem, little is known about how to solve it. Previous approaches have been mainly empirical, using very fine scale data or synthetic data from simulations, which are subsampled at various resolutions. Movement characteristics calculated at these varying sampling rates are then compared to the values based on the full data, which represent the `true' values. Some studies have fitted functions to the relationships of movement characteristics and sampling rate \citep{Pepin:2004bh,Codling:2005uo,Mills:2006tn}. These empirically obtained functions may be used to correct movement characteristics for sampling rate. While correction factors derived from movement data remain situation-specific and cannot easily be applied across species \citep{Ryan:2004et,Rowcliffe:2012hp}, we can obtain more general results by analyzing the effects of sampling rate at the level of the model \citep{Codling:2005uo,Rosser:2013ia}. Often, important characteristics of movement can be well captured by models, and therefore analyzing the properties of models can provide more general insights. However, only few such studies exist. An approach to circumvent the problem of scale-dependent statistical inference has been taken by \cite{Fleming:2014te}, who use the semivariance function of a stochastic movement process to identify multiple movement modes acting at different temporal scales. The method takes into account all possible time lags between observations. However, there are limitations as to the movement processes that can be included in this analysis \citep{Fleming:2014te}.

Here, we present a rigorous framework for studying how movement models react to changes in sampling rate, and we use this framework to analyze a class of models based on random walks. With our analysis, we seek to understand whether, and how, models can help to compensate mismatching temporal scales between different data sets or between data and behavioural process of interest. Focusing on spatially-explicit random walks, we investigate whether there are models that can validly be applied to data with different temporal resolutions and how we can account for the differences in resolutions in our interpretation of statistical inference results. In particular, we are interested in how model parameters, and their estimates, change as we decrease the temporal resolution. While estimates may change due to a shift in behavioural scale, which we always need to be aware of, we are interested in the changes that arise from the method, that is the model. Our framework is related to the statistical concept of robustness, which aims at safeguarding statistical procedures against violations of model assumptions \citep{Hampel:1986vi,Huber:2009vi}. Often, such violations refer to deviations from assumed probability distributions (e.g. Normal errors), which may result in outliers, misspecified relationships between response and explanatory variables in regression analyses, or violations of the common independence assumption. In this paper, we define robustness of movement models against changes in temporal discretization. In our framework, we treat robustness as a formal property of a model, namely the movement model. If a model has this property, it can be applied to data with varying resolutions. Additionally, while model parameters do not stay the same, they change systematically and can be translated between resolutions.

Our paper is outlined as follows. In section~\ref{s;definitions}, we define what we mean by a movement model to be robust against changes in temporal resolution. We provide three different definitions, varying in their strength of conditions. In section~\ref{s;analysis}, we present different approaches how the definitions can be used to analyze robustness of movement models. Depending on models' complexity and preexisting information, we can use formal analytical methods, numerical calculations, as well as Monte Carlo and simulation approaches. We use these approaches to examine robustness of spatially-explicit random walks and resource-selection models, and we summarize our findings in section~\ref{s;results}. In section~\ref{s;discussion}, we discuss the relevance of our robustness framework for statistical inference and also draw specific conclusions for spatially-explicit resource-selection models.

\section{Robustness of Markovian movement models}
\label{s;definitions}

We consider movement models that are discrete-time Markov processes of the form $(\bm{X}_t,\, t\in T)$, where $\bm{X}_t\in\mathbb{R}^2$ is an individual's location and $T=\{0,\tau, 2\tau, \dots\}$ is a set of regularly spaced times. This means that we assume that the time interval $\tau>0$ between two successive location measurements is fixed. Such data often arise from terrestrial animals fitted with GPS devices \citep{Frair:2010vw}. The time interval $\tau$ of the model is usually specified by the resolution of the data. We denote the one-step transition density for the probability of moving from location $\bm{y}$ to $\bm{x}$ between times $t-\tau$ and $t$ by $p_{t-\tau,t}(\bm{x}|\bm{y}, \bm{\theta})$, where $\bm{\theta}\in\bm{\Theta}$ is a vector of model parameters. This notation highlights that the transition density can be time-heterogeneous.

We consider sub-models that consist of every $n$th location of the original model for $n\in\mathbb{N}$. The transition density of the $n$th sub-model for the probability of moving from location y to x between times $t-n\tau$ and $t$ is denoted by $p_{t-n\tau,t}(\bm{x}|\bm{y}, \bm{\theta})$; compare Fig.~\ref{f;fig1}. A priori, the function $p_{t-n\tau,t}$ can have an entirely different form than $p_{t-\tau,t}$ and may correspond to a different probability distribution. However, via the Chapman-Kolmogorov equation, the $n$-step transition density can be written as a marginal density,
\begin{multline}
p_{t-n\tau,t}(\bm{x}_t|\bm{x}_{t-n\tau}, \bm{\theta})  \\
= \int_{\mathbb{R}^2\times \dots \times\mathbb{R}^2} p_{\text{joint}}(\bm{x}_t, \bm{x}_{t-\tau}, \dots, \bm{x}_{t-(n-1)\tau} | \bm{x}_{t-n\tau}, \bm{\theta}) \, d\bm{x}_{t-\tau}\dots \D \bm{x}_{t-(n-1)\tau},
\end{multline}
where we marginalize over all intermediate locations visited between times $t-n\tau$ and $t$. For simplicity, we use the general subscript `joint' to denote any joint density of multiple locations. From the notation of the locations it is clear which joint density is meant. The Markov property of the model allows us to stepwise split up the joint density as follows
\begin{multline}
p_{\text{joint}}(\bm{x}_t, \bm{x}_{t-\tau}, \dots, \bm{x}_{t-(n-1)\tau} | \bm{x}_{t-n\tau}, \bm{\theta}) \\
= p_{t-\tau,t}(\bm{x}_t|\bm{x}_{t-\tau},\bm{\theta}) \, p_{\text{joint}}(\bm{x}_{t-\tau}, \dots, \bm{x}_{t-(n-1)\tau} | \bm{x}_{t-n\tau}, \bm{\theta}).
\end{multline}
We can continue this until we obtain
\begin{multline}\label{e;n-step density}
p_{t-n\tau,t}(\bm{x}_t|\bm{x}_{t-n\tau}, \bm{\theta}) \\
= \int_{\mathbb{R}^2\times \dots \times\mathbb{R}^2} \, \prod_{k=1}^{n-1} p_{t-k\tau,t-(k-1)\tau}(\bm{x}_{t-(k-1)\tau}|\bm{x}_{t-k\tau}, \bm{\theta}) \, d\bm{x}_{t-\tau}\dots \D \bm{x}_{t-(n-1)\tau}.
\end{multline}
Therefore, we can use the one-step densities to calculate the $n$-step density; compare Fig.~\ref{f;fig1}. For statistical inference, and thus for our robustness concept, the model parameter vector $\bm{\theta}$ plays a crucial role. Although the $n$-step density may belong to a different distribution than the one-step density, equation~\eqref{e;n-step density} justifies that we use the same parameter $\bm{\theta}$ in the notation of the $n$-step density as in the one-step density.

We define robustness in terms of the one-step and $n$-step densities of a model.
\begin{definition}[Robustness of degree $n$]\label{t;definition robustness}
Let $n\in \mathbb{N}$ be finite. A movement model of the above type is \emph{robust of degree $n$} if there exists an injective function $g_n: \bm{\Theta}\rightarrow \bm{\Theta}$ such that
\begin{equation}\label{e;robustness}
 p_{t-n\tau,t}(\bm{x}|\bm{y}, \bm{\theta}) = p_{t-\tau,t}(\bm{x}|\bm{y}, g_n(\bm{\theta})) \ \text{for all}\ t\in T \ \text{and} \ \bm{x},\bm{y} \in \mathbb{R}^2.
\end{equation}
\end{definition}
This definition requires that the $n$-step densities are of the same functional form as the one-step transitions, where parameters of the model are appropriately transformed via the function $g_n$. This means that if a model is robust, the $n$th sub-model is in fact the same as the original model but with systematically adjusted parameters. The parameter transformation $g_n$ allows us to extrapolate the original parameter $\bm{\theta}$ to the coarser temporal discretization of the $n$the sub-model. Additionally, we can use the $n$th sub-model to infer the parameter $\bm{\theta}$ of the original model, because we can invert $g_n(\bm{\theta})$. Note, however, that this rests on the assumption that the original model defines the process of interest. If, instead we start at the coarser resolution, we would also need surjectivity of the function $g_n$ to conclude the existence of the finer model.

Robustness of degree $n$ has important implications. Given a behavioural process of interest, described by a robust model with parameter $\bm{\theta}$, we can apply the model not only to data with matching temporal resolution $\tau$ but also to coarser data with resolution $n\tau$ (e.g. double time interval for $n=2$). The parameter estimate $\bm{\psi}$ that we obtain from the coarser data is in fact an estimate of $g_n(\bm{\theta})$. From this, we can infer the value of $\bm{\theta}$ via $\bm{\theta} = g_n^{-1}(\bm{\psi})$. Additionally, robustness allows us to compare studies pertaining to the same behavioural process but using data sets with different resolutions. If $\bm{\theta}$ is the estimate based on the finer data, it can be extrapolated to the coarser scale via the parameter transformation $g_n(\bm{\theta})$, for all degrees $n$ for which the model is robust.

Robustness as in Definition \ref{t;definition robustness} is a strong condition that we do not expect to hold but in few special cases of the density $p_{t-\tau,t}(\bm{x}|\bm{y},\bm{\theta})$. However, equation \eqref{e;robustness} may hold up to a function $v(\bm{x},\bm{y})$, where $v$ is a bounded function that could also depend on $n$ or $\tau$. For practical applications, such \emph{approximate} or \emph{asymptotic robustness} may be sufficient. Therefore, we provide two additional definitions.
\begin{definition}[Asymptotic robustness of degree $n$]\label{t;definition asymptotic}
Let $n\in \mathbb{N}$ be finite. A movement model of the above type is said to be \emph{asymptotically robust of degree $n$} if there exists an injective function $g_n: \bm{\Theta}\rightarrow \bm{\Theta}$ and a function $v: \mathbb{R}^2\times\mathbb{R}^2\times\mathbb{R}^{+}\rightarrow \mathbb{R}^+$ with the property $v(\bm{x},\bm{y};\tau)-1=\mathcal{O}(\tau)$ on $\mathbb{R}^2\times\mathbb{R}^2 \times \mathbb{R}^+$, such that
\begin{equation}\label{e;definition asymptotic}
 p_{t-n\tau,t}(\bm{x}|\bm{y},\bm{\theta}) = p_{t-\tau,t}(\bm{x}|\bm{y}, g_n(\bm{\theta}))\, v(\bm{x},\bm{y};\tau) \ \text{for all}\ t\in T \ \text{and} \ \bm{x},\bm{y} \in \mathbb{R}^2.
\end{equation}
\end{definition}
Here, $\mathcal{O}$ denotes the Landau symbol for the order of a function. If a model is asymptotically robust, the $n$-step densities are not exactly the same as the one-step densities, as was required in Definition~\ref{t;definition robustness}. However, the discrepancy between the densities is bounded by a function that is proportional to $\tau$. More precisely, for an asymptotically robust model we have
\begin{equation}\label{e;asymptotic ratio}
1-C\tau \leq \frac{p_{t-n\tau,t}(\bm{x}|\bm{y},\bm{\theta})}{p_{t-\tau,t}(\bm{x}|\bm{y}, g_n(\bm{\theta}))} \leq 1+C\tau
\end{equation}
for all $\bm{x}$, $\bm{y}$ and $\bm{\theta}$, for some constant $C>0$. Therefore, if the time interval $\tau$ of the model is sufficiently small, the $n$-step density will closely resemble the one-step density with appropriately adjusted parameters. Asymptotic robustness of degree $n$ implies that robustness of degree $n$ is achieved as $\tau\rightarrow 0$, that is when the time interval~$\tau$ approaches zero.

In applications, the time interval $\tau$ may not be chosen sufficiently small for Definition~\ref{t;definition asymptotic} to be useful. Therefore, we give a variation of Definition~\ref{t;definition asymptotic}, in which the function $v$ does not depend on $\tau$.
\begin{definition}[Approximate robustness of magnitude $\delta$ and degree $n$]\label{t;definition approximate}
Let $n\in \mathbb{N}$ be finite. A movement model of the above type is said to be \emph{approximately robust of magnitude $\delta$ and degree $n$} if there exists an injective function $g_n: \bm{\Theta}\rightarrow \bm{\Theta}$ and a function $v: \mathbb{R}^2\times\mathbb{R}^2 \rightarrow \mathbb{R}^+$ with the property $0 < 1-\delta \leq v(\bm{x},\bm{y}) \leq 1+\delta$ for all $\bm{x}$, $\bm{y}\in \mathbb{R}^2$, for a $\delta>0$, such that
\begin{equation}\label{e;definition approximate}
 p_{t-n\tau,t}(\bm{x}|\bm{y}, \bm{\theta}) = p_{t-\tau,t}(\bm{x}|\bm{y}, g_n(\bm{\theta}))\, v(\bm{x},\bm{y}) \ \text{for all}\ t\in T \ \text{and} \ \bm{x},\bm{y} \in \mathbb{R}^2.
\end{equation}
\end{definition}
Analogously to equation~\eqref{e;asymptotic ratio}, condition~\eqref{e;definition approximate} can be written as 
\begin{equation}
1-\delta \leq \frac{p_{t-n\tau,t}(\bm{x}|\bm{y},\bm{\theta})}{p_{t-\tau,t}(\bm{x}|\bm{y}, g_n(\bm{\theta}))} \leq 1+\delta.
\end{equation}
In fact, we may consider two different types of magnitudes. Setting 
\begin{equation}
v(\bm{x},\bm{y}) := \frac{p_{t-n\tau,t}(\bm{x}|\bm{y},\bm{\theta})}{p_{t-\tau,t}(\bm{x}|\bm{y}, g_n(\bm{\theta}))},
\end{equation}
this function depends a priori on the parameters, that is we have $v(\bm{x},\bm{y};\bm{\theta})$, and the magnitude is $\delta_{\bm{\theta}}$. If $\max_{\bm{\theta}} \delta_{\bm{\theta}}$ exists, then this is the overall magnitude for the model with all possible parameter values. The magnitude determines how close $n$-step densities are to the parameter-adjusted one-step densities. If $\delta$ is small, then the correction function $v$ is close to one everywhere, and thus the $n$-step density has similar values as the one-step density over its entire domain. 

Asymptotic and approximate robustness have similar implications for inference as robustness, but only approximately. The quality of the approximation depends on $\tau$ or the magnitude $\delta$. Suppose we wish to estimate parameters of a behavioural process that we formulate in a model. Suppose we consider the time interval $\tau$ as suitable for the process. If the model is robust of degree $n$, we can use data not only at the matching scale but also at a coarser scale. For example, if the model is robust of degree 2, we can use data obtained at time interval $2\tau$. Because the model is also valid for the coarser scale, we can translate parameter estimates between the scales via the function $g_n$. If a model is asymptotically or approximately robust, the model is not exactly but still approximately valid for the coarser scale. To see this, consider the likelihood function
\begin{equation}
L_1(\bm{\theta}|\{\bm{x}_0, \bm{x}_{\tau}, \bm{x}_{2\tau}, \dots,\}) = \prod_{t\in\{\tau,2\tau,\dots\}} p_{t-\tau,t}(\bm{x}_t|\bm{x}_{t-\tau},\bm{\theta}).
\end{equation}
If a model is robust of degree $n$, the likelihood for data at time interval $n\tau$ is
\begin{equation}\label{e;likelihood}
\begin{split}
L_n(\bm{\theta}|\{\bm{x}_0, \bm{x}_{n\tau}, \bm{x}_{(n+1)\tau}, \dots,\})
&= \prod_{t\in\{n\tau,(n+1)\tau,\dots\}} p_{t-n\tau,t}(\bm{x}_t|\bm{x}_{t-n\tau},\bm{\theta}) \\
&= L_1(g_n(\bm{\theta})|\{\bm{x}_0, \bm{x}_{n\tau}, \bm{x}_{(n+1)\tau}, \dots,\}).
\end{split}
\end{equation}
If a model is asymptotically robust, we have instead
\begin{equation}
L_1(g_n(\bm{\theta})) \cdot (1-C\tau +\mathcal{O}(\tau^2)) \leq L_n(\bm{\theta}) \leq L_1(g_n(\bm{\theta})) \cdot (1+C\tau +\mathcal{O}(\tau^2)),
\end{equation}
omitting the notation of the data, which is the same as in equation~\eqref{e;likelihood}.
Analogously, for approximate robustness we have
\begin{equation}
L_1(g_n(\bm{\theta})) \cdot (1-\delta +\mathcal{O}(\delta^2)) \leq L_n(\bm{\theta}) \leq L_1(g_n(\bm{\theta})) \cdot (1+C\delta +\mathcal{O}(\delta^2)).
\end{equation}

Therefore, if a model is asymptotically or approximately robust of degree $n$, we may loosely write $L_n(\bm{\theta}) \approx L_1(g_n(\bm{\theta}))$, that is the likelihood functions based on data at time interval $\tau$ and on data at interval $n\tau$ are approximately the same. Thus, if data at time interval $\tau$ is not available, we can analyze data at time interval $n\tau$ instead, using the likelihood $L_1$ of the original model. Parameter estimates obtained in this way can be translated to the scale $\tau$ by using the inverse parameter transformation $g_n^{-1}$. Such results from statistical inference based on $L_1$ may be close to results based on the correct $L_n$, which may be difficult to compute. How close results are depends on the quality of the approximations in Definitions~\ref{t;definition asymptotic} and \ref{t;definition approximate} via $\tau$ or $\delta$. For example, if a model is approximately robust with a very small magnitude $\delta$, the likelihood $L_1$ will describe data at time interval $n\tau$ almost as well as $L_n$.

\section{Analyzing spatially-explicit random walks}
\label{s;analysis}

We used the robustness definitions to analyze spatially-explicit random walk models. These models merge general movement tendencies of an individual with decisions based on specific characteristics of locations, such as environmental features and available resources. We investigated how the models react when applied to data with increasingly coarser temporal resolution.

Our robustness definitions have two key features. First, the one-step transition densities of the model and the $n$-step densities of the sub-models need to have the same form. Second, model parameters, which are parameters of the densities, need to be transformed by a known function $g_n$. We can assume different approaches to investigate robustness properties of a model, depending on whether we have a candidate for the parameter transformation $g_n$ or not. If prior knowledge allows us to investigate robustness for a given or hypothesized parameter transformation, we can calculate and compare the $n$-step density $p_{t-n\tau,t}(\bm{x} | \bm{y}, \bm{\theta})$ and the parameter-adjusted one-step density $p_{t-\tau,t}(\bm{x} | \bm{y}, g_n(\bm{\theta}))$. By showing equality of the two densities, we can verify robustness. For complex models, analytical calculations may be difficult, or even impossible. In these cases, we may resort to numerical calculations, especially when approximate robustness is sufficient. 

In many situations, we may not know $g_n$ a priori, nor have any anticipation. Or, we may have tested robustness for a hypothesized parameter transformation but got poor results. In these cases, we need to establish some information on possible forms of the parameter transformation. Additionally, for complex models, numerical calculation of the high-dimensional integral required for the $n$-step density (compare equation~\eqref{e;n-step density}) may become inaccurate. A solution is then to draw on the ideas of Monte Carlo sampling. Monte Carlo methods and simulations are useful when probability densities are difficult to compute in closed-form but can conveniently be sampled from \citep[e.g.,][]{robert2000monte}. In the following, we demonstrate both approaches for analyzing movement models' robustness.

\subsection{Analytical and numerical approach}
\label{s;analytical approach}

Spatially-explicit random walks can be created by merging two elements in the transition density of the model. One component is the general movement kernel $k_{\bm{\theta}_1}(x;y)$, which can be the transition density of any standard random walk, describing the probability that an individual takes a step from $y$ to $x$ if there were no environmental information available. A second part of the model, given by the weighting function $w_{\bm{\theta}_2}(x)$, rates each possible step based on the location $x$. The transition densities of the full model takes the form
\begin{equation}\label{e;transition density}
p_{t-\tau,t}(x|y, \bm{\theta}_1, \bm{\theta}_2) = \frac{k_{\bm{\theta}_1}(x;y)\,w_{\bm{\theta}_2}(x)}{\int_{\mathbb{R}} k_{\bm{\theta}_1}(z;y) w_{\bm{\theta}_2}(z) \, \D z }.
\end{equation}
The integral in the denominator serves as a normalization constant.

For simplicity, we restricted our analysis to the one-dimensional case, that is we assumed that $X_t\in\mathbb{R}$. We further focused on Gaussian kernels $k_{\bm{\theta}_1}(x;y) = k_{\sigma}(x;y)$, where $k_{\sigma}(x;y)$ is a Gaussian density with mean $y$ and standard deviation $\sigma$. The weighting function $w_{\bm{\theta_2}}(x)$ was assumed to be positive everywhere to ensure that equation~\eqref{e;transition density} defines a density. In the following we simply use $\bm{\theta}$ for the parameter vector of the weighting function, or, when it is clear which parameters refer to the weighting function, we drop the subscript for the parameter in the notation of the weighting function entirely.

Note that the transition density~\eqref{e;transition density} does not depend on time explicitly. Still, as the individual moves through space over time, the centre location $y$ of the kernel shifts. Although the kernel is a function of the distance $\|x-y\|$ only, the weighting function adds a spatially explicit component. Therefore, unless the individual remains at the same location, the transition kernel effectively changes at every time step. In the following, we omit the time-related subscript in the notation of the density and simply write $p_1$ for the transition density~\eqref{e;transition density} and $p_n$ for the $n$-step density. The time interval of the original process is always assumed to be $\tau$. The distinction between one-step and $n$-step density is still important, because the $n$-step density is in fact an integral over multiple one-step densities; compare equation~\eqref{e;n-step density}.

We investigated whether we could find weighting functions $w_{\bm{\theta}}(x)$ such that the model with transition density~\eqref{e;transition density} is robust, asymptotically robust or approximately robust. We started by verifying Definition~\ref{t;definition robustness} for simple cases of the weighting function for a fixed parameter transformation $g_n$. As highlighted above, the parameter transformation is a key element, translating parameters between different temporal resolutions. For the parameter of the Gaussian movement kernel $k_{\sigma}$, we obtained a candidate for the transformation based on the linearity of the Gaussian distribution. If we only consider the kernel $k_{\sigma}$, we have a simple random walk with normally distributed steps between locations. The $n$-step density~\eqref{e;n-step density} is then the density of a sum of $n$ normally distributed random variables, which is again normal with standard deviation $\sqrt{n}\sigma$. Therefore, we assumed that the transformation of the kernel's standard deviation was given by $g_n(\sigma) = \sqrt{n}\sigma$. For the parameters of the weighting function we assumed that they remain unaffected, that is $g_n(\bm{\theta})=\bm{\theta}$. This is an ideal property for a weighting function, as it guarantees validity of inference results across different sampling rates without further translation.

In a next step, we used the same parameter transformation $g_n(\sigma, \bm{\theta}) = (\sqrt{n}\sigma, \bm{\theta})$ to establish conditions on the weighting function such that the model is asymptotically robust.  For this, we assumed that the parameter of the kernel, the standard deviation, was influenced by the time interval $\tau$, that is $\sigma = \sigma(\tau)$. This reflects that an individual may travel larger distances during longer time intervals. Because of the linearity of the Gaussian distribution, we assumed the relationship $\sigma(\tau) = \sqrt{\tau}\omega$, for some $\omega>0$. For certain conditions on the weighting function, we verified Definition~\ref{t;definition asymptotic} analytically for the robustness degree $n=2$ by calculating the function $v(x,y;\tau)$ and placing bounds on it.

As alternative to an analytical approach, we can calculate the ratio of two-step and one-step density numerically to see whether we can find a function $v(x,y;\tau)$ according to Definition~\ref{t;definition asymptotic} for the degree $n=2$. Define $\delta(\tau) :=\max_{x,y}|v(x,y;\tau)-1|$. Note that since step densities depend on $\tau$ through $\sigma(\tau)$, we may equivalently consider $\delta(\sigma)$. If this is independent of the other parameters $\bm{\theta}$, we can obtain the bound on $v$ as $\delta:=\max_{\sigma}\delta(\sigma)$, if this maximum exists. More generally, we can consider $v(x,y,\sigma,\bm{\theta})$ and calculate $\delta_{\bm{\theta}}(\sigma) := \max_{x,y}|v(x,y;\sigma,\bm{\theta})-1|$. This $\delta_{\bm{\theta}}(\sigma)$ is the magnitude of approximate robustness (degree 2) for a model with a fixed weighting function, including parameter values. An overall magnitude for the family of models consisting of the model for all parameter values can be obtained as $\delta := \max_{\sigma,\bm{\theta}}\delta_{\bm{\theta}}(\sigma)$. We demonstrate these two numerical approaches with an example weighting function.

\subsection{Simulation approach}
\label{s;simulation approach}

\subsubsection{Resource selection models}

Resource selection analyses link animal location data and environmental variables to understand animals' space-use patterns in relation to their habitat. These studies provide insight into species' preferences or avoidance of habitat characteristics, which is important information for wildlife management and conservation purposes \citep{Hebblewhite:2008tg,Latham:2011uq,Squires:2013hs}. Central methodological elements are resource selection functions (RSF) and resource selection probability functions (RSPF), describing the probability of selection of certain units (e.g. pixels of land) by an organism based on environmental covariates \citep{Manly:2002un,Boyce:2002vz,Lele:2006wc}.  RSF and RSPF have been used on their own in a mere statistical framework \citep{Boyce:2002vz,Courbin:2013hu}, incorporated into spatially-explicit models \citep{Rhodes:2005uu,Aarts:2011fz}, and become part of mechanistic movement models \citep{Moorcroft:2008wa,Potts:2014wt}). We refer to \cite{Lele:2013hw} for details about the distinction of RSF and RSPF and use RSF as a general term for both concepts, unless otherwise stated.

We include resource selection in the spatially-explicit random walk with transition density~\eqref{e;transition density} by letting the weighting function take the form of an RSF, $w_{\bm{\theta}}(x) = w_{\bm{\theta}}(\bm{r}(x))$, where $\bm{r}(x)=(r_1(x),\dots,r_n(x))$ is a vector of resource covariates at location $x$. Each $r_j$ is a function over space, representing resource covariates such as elevation, biomass measures, land cover type, and much more. The transition density becomes
\begin{equation}\label{e;general resource model}
p_1(x|y, \sigma, \bm{\theta}) = \frac{k_{\sigma}(x;y)\,w_{\bm{\theta}}(\bm{r}(x))}{\int_{\mathbb{R}} k_{\sigma}(z;y) w_{\bm{\theta}}(\bm{r}(z)) \, \D z }.
\end{equation}
In practice, geographical information is spatially discrete, and therefore the normalizing integral in equation~\eqref{e;general resource model} becomes a sum over pixels, or cells, of land. Note that we still restrict our attention to one-dimensional models.

The RSF can take various forms, and here we consider the two most commonly used ones \citep{Manly:2002un,Lele:2006wc}, the exponential RSF, 
\begin{align}
w_{\exp}(\bm{r}(x)) &= \exp\left(\bm{\beta}\cdot \bm{r}(x) \right)  \label{e;exponential rsf} \\
\intertext{and the logistic function,}
w_{\log}(\bm{r}(x)) &= \frac{\exp\left( \alpha + \bm{\beta}\cdot \bm{r}(x) \right)}{1+\exp\left( \alpha + \bm{\beta}\cdot \bm{r}(x) \right)}.  \label{e;logistic rsf} 
\end{align}
The vector $\bm{\beta}$ comprises all selection parameters with respect to resource covariates $\bm{r}$. A higher selection parameter means stronger selection with respect to the corresponding resource. In the logistic form, $\alpha$ is an intercept parameter, which can shift the inflection point of the logistic function away from zero. The inflection point is the point where the logistic function attains a value of 0.5, that is where the probability of selecting a resource is 50\%. If the exponential form~\eqref{e;exponential rsf} is used, an intercept similarly to the one used in equation~\eqref{e;logistic rsf} is not identifiable, because it cancels in the definition of the transition density~\eqref{e;general resource model}. Therefore we have omitted it in equation~\eqref{e;exponential rsf}. The function $w_{\text{log}}$ has range $(0,1)$ and can therefore be used to describe probabilities. This means that this form can be used as RSPF, which for a given location $y$ specifies the probability that an animal selects this location, given the covariate values of the location. In contrast, the exponential RSF can only specify values proportional to this probability, with unknown proportionality constant \citep{Lele:2013hw}.

\subsubsection{Sampling models and sub-models}

We examined the two models with weighting functions $w_{\exp}$ and $w_{\log}$ for their robustness. Because the weighting functions depend on space through environmental information $\bm{r}$ they are highly non-linear, and therefore the transition densities are difficult to examine analytically. Sampling probability distributions is a convenient work around and has the additional advantage that we can control parameters and isolate processes of interest. We thus simulated sample trajectories from the model with transition densities~\eqref{e;general resource model}. The joint density of a movement trajectory $(x_1, \dots, x_N) \in \mathbb{R}^N$ of length $N\in\mathbb{N}$ is given by
\begin{equation}\label{e;trajectory probability}
p_{\text{joint}}(x_1, \dots, x_N, \bm{\theta}) = p_1(x_1, \bm{\theta}) \prod_{t=2}^{N} p_1(x_t | x_{t-1}, \bm{\theta}).
\end{equation}
Thus, we sampled successively from the transition densities to obtain a full movement trajectory. We obtained samples from the subprocess $\bm{x}_n = (x_1, x_{n+1}, \dots)$ consisting of every $n$th location by subsampling the full trajectories. These subsamples represent samples from the model with transition densities being the $n$-step densities $p_n(\cdot | \cdot,\bm{\theta})$. 

Because the models rely on environmental data, we simulated resource landscapes as realizations of Gaussian random fields with exponential covariance model \citep{Haran:2011vv,RandomFields:2013}. This resulted in spatially correlated resource landscapes, thus ensuring realism; compare Figure~\ref{f;lands} in Appendix~\ref{s;supp figures}. The sampled movement trajectories were based on these simulated landscapes. To avoid confounding effects and to keep results as clear as possible, we assumed that the weighting function was based on only one resource $r$, thus we have $w_{\bm{\theta}}(r(x))$. With the exponential covariance model, we assumed that the covariance of resource values at two different locations is given by 
\begin{equation}
\text{Cov}(r(x), r(y)) = \exp\Bigl(\frac{|x-y|}{s}\Bigr),
\end{equation}
where $s$ affects the decrease of the spatial autocorrelation with increasing distance.

We sampled trajectories for varying parameter values. We used $\sigma\in\{5,6,7\}$  and $\beta\in\{0.5,1,1.5,2\}$ in all combinations. In the model with logistic RSF $w_{\log}$, we further combined the values $\alpha\in\{-1,-0.5,0,0.5,1\}$ with all other parameters. For each parameter combination, we sampled 16 trajectories for 15,000 time steps each; compare Fig.~\ref{f;traj exp},\ref{f;traj log} in Appendix~\ref{s;supp figures}. For each of the 16 trajectories, we used a different resource landscape, repeating the same set of resource landscapes across different parameter combinations. The 16 landscapes were generated with varying spatial autocorrelation, $s$ ranging between 200--500. This led to a total of 192 sampled trajectories for the model with exponential RSF and 960 trajectories for the model with logistic RSF. We subsample every trajectory at levels $n=1, \dots, 15$, leaving 1000 steps for the coarsest time series. The subsample for $n=1$ is the original trajectory.

\subsubsection{Analyzing parameters}

While the simulated trajectories represent samples from the original model with transition densities $p_1(\cdot|\cdot,\bm{\theta})$, the subsamples of the full trajectories provide us with samples from the sub-models with $n$-step densities $p_n(\cdot|\cdot,\bm{\theta})$. To learn about the model's robustness properties, we need to test whether the subsamples reconcile with the parameter-adjusted one-step densities $p_1(\cdot|\cdot,g_n(\bm{\theta}))$ for some parameter transformation $g_n$. For a given parameter transformation, we can achieve this by analyzing the fit of the model with transitions $p_1(\cdot|\cdot,g_n(\bm{\theta}))$ with the subsamples. When $g_n$ is unknown, or when the fit for a hypothesized $g_n$ is poor, we first need to investigate the behaviour of the parameters under subsampling to see whether we can find a function $g_n$ as required by our robustness definitions.

Here, we both tested a priori expectations on the parameter transformation and searched for better alternatives. We estimated parameters for all trajectories and their subsamples using maximum likelihood optimization. The likelihood function for the full trajectories is given in equation~\eqref{e;trajectory probability}. For subsamples, we applied the same model, although we did not know whether subsamples of trajectories followed the same (parameter-adjusted) process as full trajectories. We expected parameter estimates for the full trajectories to be close to the values that we used during the simulations. We call these the `true values', although deviations in the simulations are possible, because simulated trajectories are realizations of stochastic processes. Our main interest are parameter estimates for the subsamples. To distinguish estimates from underlying true parameters, we denote the estimate with a hat, e.g. $\hat{\sigma}$. Ideally, the parameters of the subsamples should follow some function $g_n(\sigma, \alpha,\beta)$, and so should the estimates. To see whether such a function exists, we fitted non-linear regression models to the relationship of parameter estimates of subsamples and the subsampling amount $n$. For each parameter, we fitted two models. One model was more restrictive and represented a priori expectations, whereas the other model had an additional free parameter that allowed more flexibility for the parameter transformation.

The general movement kernel $k$ has one parameter, the standard deviation $\sigma$ of the Gaussian distribution. This kernel describes the general movement tendencies of the animal, and $\sigma$ influences the distance covered in each step. With increasing subsampling, the temporal resolution of the movement path becomes coarser, and we thus expected the standard deviation of the kernel to increase. Each step in a subsample is in fact the accumulated result of one or several steps in the full trajectory. If the kernel is the only force driving the movement, the linearity of the Gaussian distribution caused us to expect the standard deviation of the kernel to increase as $\sqrt{n}\sigma$; compare section~\ref{s;analytical approach}. With additional resource selection, however, there may be deviations from this behaviour.

For the resource selection parameters $\alpha$ and $\beta$, an ideal behaviour would be that they remain unaffected by the subsampling, analogously to our assumptions in section~\ref{s;analytical approach}. In our model, we assume that each step is influenced by the RSF. One of the underlying assumptions of a traditional RSF is that it gives weights to locations independently of the values of other locations, which means each location is weighted by its present resource only, without consideration of alternative locations. Therefore, resource selection parameters should be independent of the temporal resolution of the data.  However, within the spatially-explicit movement framework, resource selection always occurs in the context of the current location and the available surrounding area as defined by the general movement kernel. Therefore, a change in the movement kernel due to increased subsampling may be accompanied by a change in resource selection parameters.

We fitted the non-linear regression models to the parameter estimates separately for each parameter combination. This means that in each regression, we fitted estimates of 16 trajectories and their subsamples. Because of our previous considerations about the kernel parameter $\sigma$, we assumed a power relationship between the estimate $\hat{\sigma}$ and the subsampling amount $n$, stratified by trajectories. We chose the stratification because trajectories were simulated on different landscapes. Also, for the resource selection parameters, especially when their true values were close to zero, estimates could vary between being positive and negative. In these cases, the stratification allowed for flexibility. The model for the estimate of the $n$th subsample of trajectory $i$ is 
\begin{equation}\label{e;sigma model}
\hat{\sigma}_{i,n} = \hat{\sigma}_{i,1} \cdot n^b + \varepsilon, \qquad 1\leq n \leq 15, \quad 1\leq i \leq 16,
\end{equation}
where the error term $\epsilon$ is normally distributed with mean zero and positive standard deviation $\zeta$. The maximum likelihood estimate of $b$ should ideally be close to 0.5, however as noted above, it may deviate from this value because of resource-selection mechanisms. To test whether $b$ differs from 0.5, we used model selection via AIC between the model in equation~\eqref{e;sigma model} and the model in which we fixed $b=0.5$.
 
Model choice for the resource selection parameters was less clear. Visual inspection of the estimates, preliminary fits with varying models and inspection of residuals suggested a power law for the parameter $\beta$ as well. We thus fitted the following model,
\begin{equation}\label{e;beta model}
\hat{\beta}_{i,n} = \hat{\beta}_{i,1} \cdot n^b + \varepsilon, \qquad 1\leq n \leq 15, \quad 1\leq i \leq 16.
\end{equation}
We compared the fit of this model with the model in which we assumed that subsampling does not change the estimate by setting $b=0$.

For the intercept parameter $\alpha$ in the logistic form of the resource selection function, we chose a linear model,
\begin{equation}\label{e;alpha model}
\hat{\alpha}_{i,n} = \hat{\alpha}_{i,1} + b\,(n-1) + \varepsilon, \qquad 1\leq n \leq 15, \quad 1\leq i \leq 16.
\end{equation}
Inspection of residuals suggested that in some cases the relationship between $\hat{\alpha}$ and $n$ was non-linear. However, a power-law model or other non-linear relationships were not consistently more suitable either. Therefore we remained with the simpler, the linear, model, noting that this is a mainly illustrative analysis.

\subsubsection{Calculating approximate robustness}

To accompany the simulation analysis, we examined approximate robustness properties of the two models with exponential and logistic RSF. We focused on approximate robustness of degree 2, and we tested the ideal parameter transformations $g_2(\sigma, \beta)=(\sqrt{2}\sigma, \beta)$ and $g_2(\sigma,\alpha,\beta)=(\sqrt{2}\sigma, \alpha,\beta)$ for $w_{\exp}$ and $w_{\log}$, respectively. We numerically calculated a magnitude $\delta = \max_{x,y}(|v(x,y)-1|)$ for every possible scenario that we used in the previous section. This means that we calculated a magnitude for each combination of the parameters $\sigma$, $\beta$, and $\alpha$ (in case of the logistic RSF) and for each of the 16 simulated resource landscapes. We may therefore think of $\delta$ as $\delta(\sigma, \alpha, \beta, i)$, for $1\leq i \leq 16$; compare Fig.~\ref{f;fig2} We examined whether magnitudes were influenced by parameter values and specific characteristics of the landscapes, such as their spatial autocorrelation and their overall variation $\text{Var}(r(x))$ over the spatial domain. We further calculated an overall maximum $\max_{\sigma,\alpha,\beta,i}\delta(\sigma,\alpha,\beta,i)$. We compared results between the model with exponential RSF, $w_{\exp}$, and logistic RSF, $w_{\log}$.

\section{Results}
\label{s;results}

\subsection{Analytical and numerical results}
\label{s;exact results}

We found few special cases of weighting functions $w_{\bm{\theta}}$ that, together with the Gaussian kernel $k_{\sigma}$, resulted in a robust movement model according to Definition~\ref{t;definition robustness}.

The simplest case was a constant weighting function. Such a weighting function reduces equation~\eqref{e;transition density} to the case of a homogeneous environment, where only general movement tendencies play a role, but no environmental information. The model is then a simple random walk with normally distributed steps between locations. Because of the linearity of the normal distribution, the model is robust of degree $n$ for all $n\in\mathbb{N}$ for the assumed parameter transformation $g_n(\sigma) = \sqrt{n}\sigma$; compare also Theorem~\ref{t;exp w} for parameters $a=b=0$.

A natural next step was to consider a linear weighting function. However, a linear weighting function violates the assumption of being strictly positive everywhere. If in equation~\eqref{e;transition density} the current location $y$ is the point at which $w$ becomes zero, the normalization integral vanishes. Also, equation~\eqref{e;transition density} can become negative and thus cease to be a valid density function. Still, we could draw on the linearity of the expectation of a random variable to look into this further. The normalization constant in the transition density~\eqref{e;transition density} can be viewed as an expectation of the form $E(w(Z))$ for a normally distributed random variable $Z$ with mean $y$. Therefore, if the function $w$ is linear, the normalization constant reduces to $w(y)$. Equation~\eqref{e;transition density} then becomes
\begin{equation} 
p_1(x|y, \sigma, \bm{\theta})  = k_{\sigma}(x;y) 
\frac{w_{\bm{\theta}}(x)}{w_{\bm{\theta}}(y) }.
\end{equation}
The right-hand side of the equation is positive whenever $x$ and $y$ are either both negative or both positive. If movement only occurs in the domain where the weighting function is positive the model is robustness within this domain. The details of the proof can be found in Appendix~\ref{s;proofs exact}.

\begin{theorem}[Linear weighting function]
\label{t;linear w}
Let $w$ be a linear function $w(x) = ax+b$, for $a, b\in \mathbb{R}$. Let $\mathcal{I}\subset\mathbb{R}$ be the interval where $w>0$. For the restricted domain $\mathcal{I}$, the movement model with transition densities~\eqref{e;transition density} is robust of degree $n$ for all $n\in\mathbb{N}$. The parameter transformation is given by $g_n(\sigma, a, b) = (\sqrt{n}\sigma, a, b)$.
\end{theorem}

We found another special case to be given by an exponential weighting function. Here, no restriction on the domain is necessary. Again, see Appendix~\ref{s;proofs exact} for details of the proof. 

\begin{theorem}[Exponential weighting function]
\label{t;exp w}
Let $w$ be an exponential function of the form $w(x) = Ce^{ax+b}$ for $C, a, b \in \mathbb{R}$.  Then the movement model with transition densities~\eqref{e;transition density} is robust of degree $n$ for all $n\in\mathbb{N}$ with parameter transformation $g_n(\sigma, C, a, b) = (\sqrt{n}\sigma, C, a, b)$.
\end{theorem}

The above two Theorems show that it is possible to verify exact robustness with the ideal parameter transformation $g_n(\sigma, \bm{\theta}) = (\sqrt{n}\sigma,\bm{\theta})$ for certain weighting functions. However, the cases are very restrictive, and robustness will fail for many other, and especially more complex, weighting functions.

We could additionally establish asymptotic robustness for more general conditions on the weighting function. The main result is summarized in the following theorem. For a detailed proof of the theorem, see Appendix~\ref{s;proof asymptotic}.

\begin{theorem}[Asymptotic robustness of degree 2]
\label{t;asymptotic robustness}
Let $w_{\bm{\theta}}$ be continuous and bounded away from zero. Let $w_{\bm{\theta}}$ further be twice differentiable with bounded second derivative. Then the model with transition densities~\eqref{e;transition density}  is asymptotically robust of degree 2 with parameter transformation $g_2(\sigma, \bm{\theta}) = (\sqrt{2}\sigma,\bm{\theta})$.
\end{theorem}

Thus, if the weighting function is well-behaved according to the theorem, we can place a bound on the factor by which the one- and two-step density vary; compare equation~\eqref{e;asymptotic ratio}. This bound is of order $\tau$, such that the discrepancy between one- and two-step density decreases with the time interval.

\begin{example}[Asymptotic robustness of degree 2]
As a simple example, consider the weighting function $w(x) = \sin(\alpha x) + \beta$ for $\alpha>0$ and $\beta>1$. The choice of $\beta$ guarantees that the weighting function is positive everywhere. The function $w$ is bounded between $0 < \beta-1 \leq w(x) \leq \beta+1$ for all $x\in\mathbb{R}$, and its second derivative is bounded by $|w''(x)| = \alpha^2$. Therefore, Theorem~\ref{t;asymptotic robustness} holds. 
\end{example}

The proof of Theorem~\ref{t;asymptotic robustness} is constructive in the sense that it provides us with a constant $C$ for equation~\eqref{e;asymptotic ratio} in terms of the bounds on $w$ and $w''$. However, this constant may be rather large and does not necessarily provide the closest bound on the function~$v$. Therefore, it can be informative to calculate approximate robustness numerically.

\begin{example}[Approximate robustness of degree 2]
\label{t;example approximate}
We continue the above example with weighting function $w(x) = \sin(\alpha x) + \beta$ for $\alpha>0$ and $\beta>1$. We calculated the function $v(x,y;\sigma, 
\alpha,\beta)$ from Definition~\ref{t;definition approximate} numerically, using different values of $\alpha$ and $\beta$ (Fig.~\ref{f;fig3}a).
From this, we obtained $\delta_{\alpha,\beta}(\sigma)$ (Fig.~\ref{f;fig3}b), which is the magnitude of approximate robustness (degree 2) for the model with specific weighting function (i.e.\ with specific parameters); compare Fig.~\ref{f;fig2}. In each case, after reaching a maximum the function vanishes for increasing $\sigma$. Therefore it appears that we can find $\delta_{\alpha,\beta} := \max_{\sigma} \delta_{\alpha,\beta}(\sigma)$. The wavelength of the sine curve, determined by $\alpha$, and the intercept $\beta$ have different effects on the function $\delta_{\alpha,\beta}(\sigma)$. While $\alpha$ shifts the curve, $\beta$ changes the height of the peak (Fig.~\ref{f;fig3}b). Therefore, it appears that $\delta_{\alpha, \beta}$ is independent of $\alpha$ and decreases for larger $\beta$. For the weighting function to be positive, $\beta$ needs to be larger than one. For $\beta=1$, the function $\delta_{\alpha, \beta}$ has a maximum at one. From these considerations, we can conclude that $\max_{\alpha,\beta} \delta_{\alpha,\beta} = 1$. This is the overall magnitude of approximate robustness (degree 2) for the family of weighting functions $w(x) = \sin(\alpha x) + \beta$, $\alpha>0$, $\beta>1$; compare Fig.~\ref{f;fig2} As a word of caution, we note that we only calculated $\delta_{\alpha,\beta}$ for a fixed number of parameter values and only within finite intervals for $x$ and $y$, and therefore results may be limited to these ranges.
\end{example}

In the region where $\delta(\sigma)$ peaks, the approximation of the parameter-adjusted one-step density $p_1(x | y, \sqrt{2}\sigma, \alpha, \beta)$ to the actual two-step density $p_2(x | y, \sigma, \alpha, \beta)$ is only rough. However, for larger values of $\sigma$, and independent of $\alpha$ and $\beta$, the function $\delta_{\alpha,\beta}(\sigma)$ seems to vanish, which means that the approximation is good and the discrepancy between two- and one-step densities may be neglected. From Theorem~\ref{t;asymptotic robustness}, we would have been able to conclude that $\delta_{\alpha,\beta}(\sigma(\tau))$ is bounded by $C\tau$, for a constant $C>0$, for all $\alpha>0$ and $\beta>1$. As we can see from the steep initial slope of $\delta_{\alpha,\beta}(\sigma)$, especially for higher values of $\alpha$, the constant $C$ would need to be rather large (Fig.~\ref{f;fig3}b). The calculations of approximate robustness could additionally show that the bound on $v(x,y)$ is in fact much smaller.

\subsection{Simulation results}
\label{s;simulation results}

\subsubsection{Results for parameter estimates}

When analyzing parameter estimates from the simulated trajectories and their subsamples, we found a difference in the behaviour of parameters between the exponential and the logistic form of the RSF. Generally, subsampling had less effect on the value of parameter estimates using the logistic form, and the behaviour of estimates agreed closer with our expectations.

For both RSF, estimates $\hat{\sigma}$ showed a good fit with the power-law model. When we used the exponential RSF, the estimated power $b$ ranged from 0.45 to 0.5 for varying parameter combinations, thus deviating from expected behaviour for some parameter combinations (Fig.~\ref{f;fig4}a). For small selection parameter $\beta$, the estimate $\hat{\sigma}$ showed the expected increase as $\hat{\sigma}\sqrt{n}$. With increasingly strong selection, i.e.\ higher value of $\beta$, estimates $\hat{\sigma}$ became smaller with increased subsampling relative to the ideal relationship. An increase in $\sigma$ did not influence the fit other than leading to a larger residual standard error $\hat{\zeta}$, which is to be expected because of the overall larger values of the dependant variable. In contrast, when using the logistic RSF, the estimated power $b$ differed only very slightly from 0.5 and in some cases, the simpler model with fixed $b$ was preferred by model selection right away (Fig.~\ref{f;fig4}b). 

The behaviour of the resource-selection parameter $\beta$ also differed between exponential and logistic RSF. For the exponential RSF, $\hat{\beta}$ showed a clear increase with increased subsampling, fitted well by our power-law model (Fig.~\ref{f;fig5}a). The power $b$ remained similar (ranging 0.105--0.124) across parameter combinations, increasing slightly with larger $\sigma$ (Fig.~\ref{f;fig5}b). For the logistic RSF, estimates $\hat{\beta}$ generally remained closer to the original values for $n=1$ (Fig.~\ref{f;fig5}c,d). In most cases, model selection via AIC preferred the power-law model to the ideal constant relationship, however, the estimated values of the power $b$ are small, with 53 out of 60 values being below 0.1 (total range 0--0.156, with one exceptional negative value $b=-0.041$). There was a tendency of $b$ to be smaller and more concentrated under stronger selection (Fig.~\ref{f;fig5}d).

Estimates of the intercept $\alpha$ in the logistic RSF showed a slight decline with increased subsampling in most cases (Fig.~\ref{f;fig6}). This decreasing trend existed no matter whether $\alpha$ was positive, negative, or zero. In general, slopes of the linear fit were all close to zero (ranging -0.047--0.058), and in a few cases the null model with $b=0$ was chosen. We found a trend in the realized intercept values in the simulated trajectories. With stronger effect of selection (larger $\beta$), the intercept estimate $\hat{\alpha}$ of original trajectories ($n=1$) was stronger concentrated around the true underlying value, which subsequently lead to a stronger concentration of estimates of subsamples (Fig.~\ref{f;fig6}).

\subsubsection{Results about approximate robustness}

When comparing magnitudes $\delta(\sigma,\alpha,\beta,i)$ of approximate robustness (degree 2) between the two models with exponential and logistic RSF, we found lower magnitudes for the model with logistic function $w_{\log}$. Magnitudes for the model with exponential RSF ranged between 0.067 and 1.82, whereas those for the model with logistic RSF ranged between 0.02 and 1.19. The 5\% quantile, the median and the 0.95\% quantile were $[0.092, 0.34, 0.97]$ (exponential RSF) and $[0.046, 0.21, 0.64]$ (logistic RSF).

We found that especially the selection parameter $\beta$ had a strong influence on magnitudes, higher values of $\beta$ leading to higher magnitudes (Fig.~\ref{f;fig7}). For the model with exponential RSF, there was a tendency that weighting functions whose underlying landscapes had higher variation $\text{Var}(r(x))$ lead to smaller magnitudes (Fig.~\ref{f;fig7}a). However, we did not find an effect of the parameter $s$ that was used in the simulations to influence the spatial autocorrelation of the landscapes. The model with logistic RSF did not show such an effect of landscape variation. The logistic model had the additional intercept parameter $\alpha$. We found that higher values of $\alpha$ tended to result in lower magnitudes (Fig.~\ref{f;fig7}b).

\section{Discussion}
\label{s;discussion}

We have proposed a new rigorous framework for analyzing movement models' capacities to compensate for varying temporal discretization of data. Our framework comprises three definitions of varying strength for robustness of discrete-time movement models. Generally, if a model is robust, it can overcome problems of mismatching temporal scales between different data sets or between data and biological questions. Because our robustness is a very strong condition that holds only for very few and generally more simple models, we  have introduced the additional concepts of asymptotic and, most importantly, approximate robustness. While for many movement models it is difficult, or even impossible, to examine the transition densities and their marginals analytically, approximate robustness properties of a model can be calculated numerically also for analytically intractable models. Therefore, we believe that especially approximate robustness will prove a useful new concept for movement analyses.

We have formulated our robustness definitions in terms of the transition densities of Markov models,  because these models are often fitted to movement data with likelihood-based methods of statistical inference. For the considered models, we can obtain the likelihood function by multiplying the transition densities of subsequent steps. If a model is robust, the transition densities keep their functional form across varying temporal scales, and parameters are transformed via a well-defined function $g_n$. The likelihood function therefore remains the same but will yield different parameter estimates. However, if the function parameter transformation is known, estimates from one scale can be translated to estimates at another scales. If a model is only approximately robust, the likelihood function will not remain exactly but at least approximately the same under a change of scale. Depending on the magnitude of the approximate robustness, the approximation of the likelihood function may be sufficiently good to allow parameter estimates to be reasonably comparable for different scales, especially if the difference in scales is small.

Our concept of robustness for discrete-time movement models is related to the formal concept of robustness in statistics. Generally speaking, robust methods in statistics acknowledge that models are approximations to reality and seek to protect outcomes of statistical procedures (e.g. hypothesis testing, estimation) against deviations from the underlying model assumptions. Classic examples are the arithmetic mean and median as estimates of a population mean: while the median is robust against outliers the mean is not \citep[e.g.][]{Hampel:1986vi}. Often, robustness is viewed in the context of deviations from assumed probability distributions \citep[distributional robustness; e.g.][]{Huber:2009vi}. For example, data may be contaminated by few observations with heavier tailed distribution than the majority of the observations. In regression analyses, robustness may also relate to the homoscedasticity assumption or the functional form of the response function \citep{Wiens:2000cq,Wilcox:2012wh}. Additionally, robustness has been considered when the assumption of independence is violated and instead observations are correlated \citep{Hampel:1986vi,Wiens:1996dm}. In our paper, we consider robustness in the context of discrete-time movement models with respect to assumptions about the temporal discretization. In view of statistical robustness, we study violations against the assumption that the temporal resolution of our movement model, a stochastic process, matches the resolution of the data, when in fact the data is only a subsample of the assumed process.

There is also a difference between our robustness of movement models and the well-established robustness in statistics. In our framework, robustness is a direct property of a model. In contrast, classical robustness in statistics is defined for objects such as estimators, test-statistics, or more generally, functionals (real-valued functions of distributions) \citep{Hampel:1971gy,Hampel:1986vi}. For the type of models we have considered here, parameter estimates cannot be obtained analytically but through numerical optimization of the likelihood function. The likelihood function is build by the model's transition densities, and thus we have defined robustness at a very basic level. A possibility for future research is to investigate whether some of the formal concepts of statistical robustness can be applied to our framework to add further insight. With our paper, we provide a new perspective for studying effects of temporal discretization of movement processes, and we hope to encourage further research.

Our analytical investigations indicate that robustness is a rare property among movement models, especially when behavioural mechanisms such as resource selection are added. Therefore, if we apply models to data without considering this issue, we are in danger of misinterpreting results and drawing erroneous conclusions. However, our analysis also shows positive prospects with respect to approximate robustness. Theorem~\ref{t;linear w} suggests that in slowly varying environments that produce locally linear weighting functions we may find some robustness. Theorem~\ref{t;asymptotic robustness} and the following examples show that certain smoothness and boundedness conditions on the weighting function can lead to approximate robustness. In addition, Example~\ref{t;example approximate} further demonstrates that approximate robustness can be investigated numerically on a case-by-case basis. We have illustrated this with a smooth weighting function $w(x)$ that is a direct function of space. In data applications, an animal's preferences for locations usually do not depend on space per se but rather through the type of habitat and available resources, and the weighting function will be less regular. In our simulation study, we have therefore presented a case with a more realistic model.

While it is difficult to analyze the transition densities and resulting $n$-step densities with analytical calculations, we have demonstrated with the simulation approach how we can still investigate robustness properties of complex models. Sampling from probability distributions instead of calculating them is the key idea of Monte Carlo methods. We have used this method to examine sub-models that have the $n$-step densities as transition densities. With this we obtained the parameter transformation $g_n$. Our approach differs from previous studies that have used subsamples of fine-scale data to establish an empirical relationship between sampling interval and movement characteristics \citep{Pepin:2004bh,Ryan:2004et,Rowcliffe:2012hp}. When using data, it can be difficult to tease apart effects that result from the methodology and effects of actual behavioural changes at different scales. Analyzing model properties as we have proposed here allows us to identify those effects of temporal discretization that are attributable to the methodology.

In our demonstration of the simulation approach, we analyzed spatially-explicit resource selection models. These models have an advantage over traditional resource-selection and step-selection functions. In the traditional, regression-type approach, observed movement steps are compared to potential steps that are obtained separately from an empirical movement kernel \citep{Fortin:2005uk,Forester:2009wa}. In this approach, movement and resource-selection are treated independently, although it is highly likely that both influence each other. In contrast, when fitting the full random walk with resource selection to data by using the likelihood function~\eqref{e;trajectory probability}, we can simultaneously estimate parameters both of the general movement kernel and the weighting function, that is the RSF.

In our analysis of the resource-selection model, we observed systematic trends in values of parameter estimates with changing temporal discretization of movement trajectories. The main purpose was not to analyze these relationships in full detail but to illustrate that they occur and thus must not be neglected. Comparing the exponential and logistic form of the spatially-explicit resource selection model, we found that estimates varied more with increased subsampling when the exponential RSF was used, compared to the logistic RSF. Using the exponential RSF, estimates of the kernel standard deviation $\sigma$ decreased with increased subsampling compared to the ideal relationship $\sqrt{n}\sigma$. On the other hand, using the logistic RSF, $\sigma$ followed the ideal relationship that would occur for a purely Gaussian process very closely, even under additional influence of resource selection. The estimated strength of resource selection, indicated by $\beta$, increased with the subsampling amount. While this effect was strongly pronounced for the model with exponential RSF, it was only weak for the logistic RSF. Therefore, if using the logistic RSF, one may expect to obtain similar inference results across varying temporal discretization.

When we calculated the magnitudes of approximate robustness for the models used in the simulations, we found that those were in line with the results for the parameter estimates. Overall, the model with logistic RSF had better robustness properties than the model with exponential RSF. We also found a matching trend for the movement parameter $\sigma$ with varying true values of $\beta$. Estimates of $\sigma$ were closer to the expected behaviour for weaker resource-selection parameters. This was also reflected in magnitudes of approximate robustness. If selection was weaker in the original model, the model exhibited better robustness properties. These results suggest that numerical calculations of approximate robustness can assist our expectations about changes in parameter estimates. On the other hand, although parameter estimates of the weighting function showed a clear difference in behaviour when comparing between the exponential and logistic RSF, differences within one model between different parameter combinations were less clear. More  analyses would be required to entangle more detailed effects.

Overall, the results from the simulations suggest that depending on the resolution of movement data, we may misinterpret animals' movement tendencies and also may overestimate resource selection effects. It is therefore important that we are aware of the changes to statistical inference that can arise merely from the methodology. Here, we have seen that changes in inference results were stronger for the resource selection model with exponential RSF compared to the logistic RSF. A possible explanation may be the additional intercept in the logistic RSF. With increased subsampling, estimates of $\alpha$ tended to decrease, possibly counteracting the increase in estimates $\hat{\beta}$. This could have led to more stability for the parameter $\sigma$ of the general movement kernel. However, this may not explain why resource selection parameters generally varied less themselves compared to the exponential RSF. Another possibility is that the different form of the RSFs causes their different behaviour. While the exponential form of the RSF greatly enhances differences in landscape values, the logistic RSF is restricted to values in the interval $(0,1)$. Theorem~\ref{t;asymptotic robustness} suggests that variation in the rate of change of the weighting function influences robustness properties. Thus the logistic RSF may produce more stable inference results for varying temporal resolutions. \cite{Lele:2006wc} suggested several alternatives to the exponential RSF. Our study case showed that the choice of resource selection functions can have implications for statistical inference and we encourage to choose resource selection functions more deliberately.

With our study we have illustrated that the concept of the parameter transformation $g_n$ can help to bridge the gap between different temporal resolutions of data. In the model with exponential RSF, we found that with increased subsampling estimates of the resource selection parameter $\beta$ deviated strongly from the original values. However, the increase in $\hat{\beta}$ could be fitted with a power-relationship. Thus, using the idea of Monte Carlo sampling, we were able to obtain a parameter transformation $g_n$. Using such transformations when comparing results obtained from data with different temporal resolutions could greatly improve our statistical inference, leading to a better understanding of movement behaviour.


\clearpage

\begin{figure}
\includegraphics{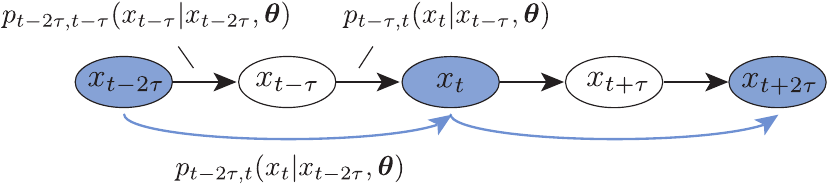}
\caption{The second sub-model consists of every second location. The transition densities of the sub-model, which we refer to as 2-step densities, are the marginals over the two intermediate one-step densities of the original model}
\label{f;fig1}
\end{figure}

\begin{figure}
\includegraphics{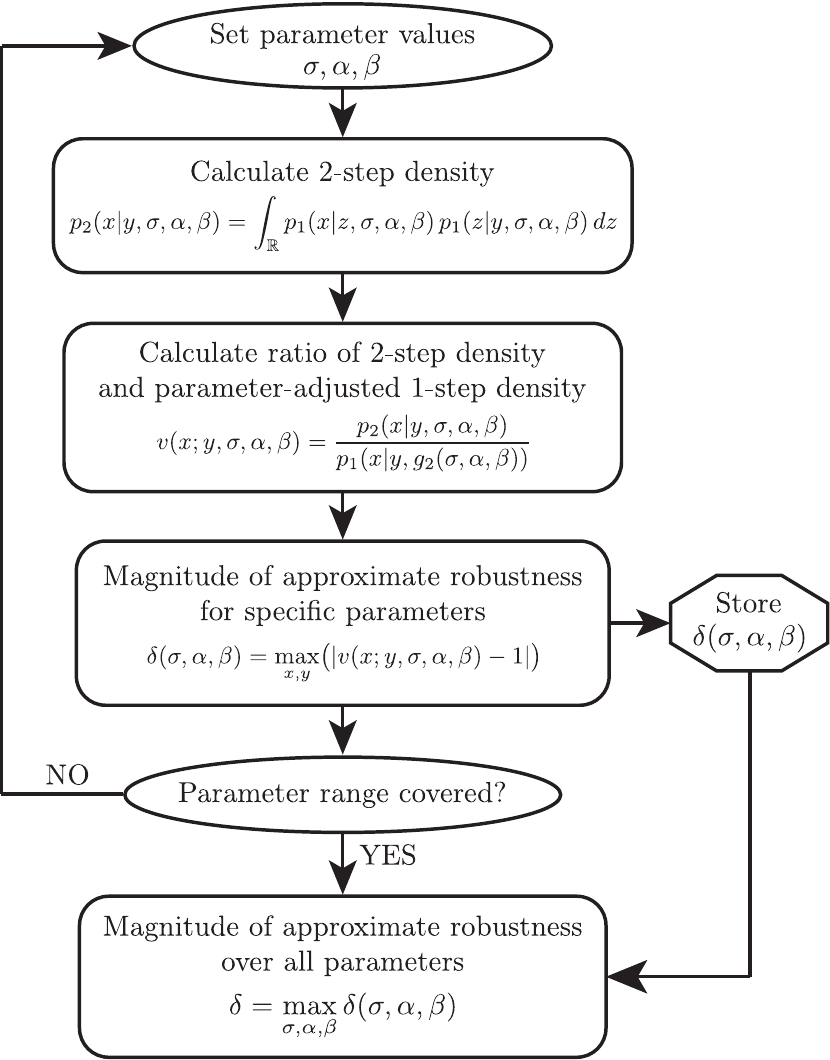}
\caption{Steps for calculating the magnitude of approximate robustness of degree 2 for a given model, where $\sigma$ is the parameter of the movement kernel, and $\alpha$ and $\beta$ are parameters of the weighting function. The one-step density $p_1$ can, for example, be equation~\eqref{e;transition density} with the weighting function from Example~\ref{t;example approximate}, or the resource selection model \eqref{e;general resource model} with weighting function~\eqref{e;exponential rsf}  or \eqref{e;logistic rsf}. When the resource selection model is used, the flowchart shows the calculation of the magnitude for one specific resource landscape $r(x)$. When calculating an overall magnitude, practically we do this for a subset of the parameter space}
\label{f;fig2}
\end{figure}

\begin{figure}
\includegraphics{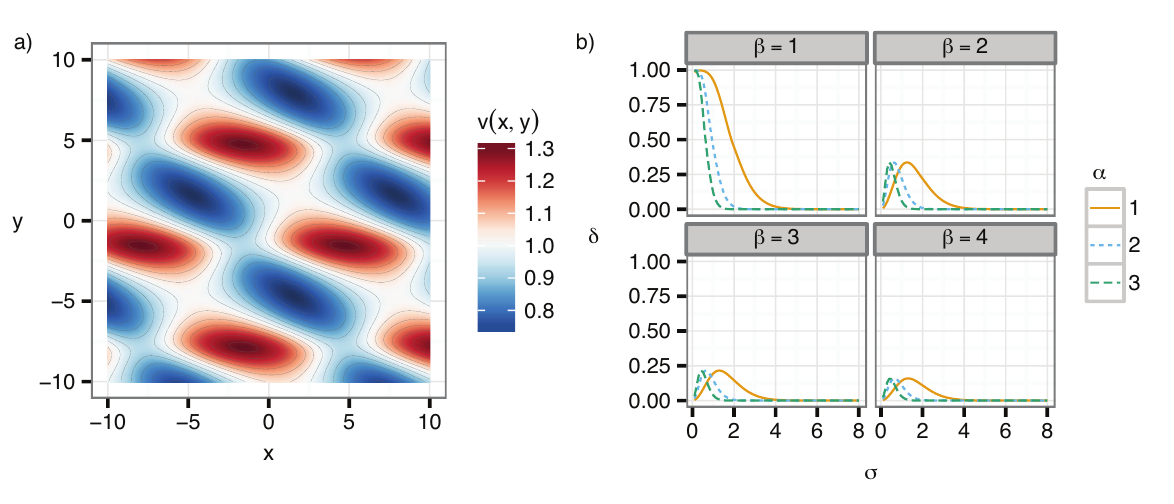}
\caption{Panel a): Numerical calculation of the function $v(x,y)$, which is the ratio of two-step density $p_{t-2\tau,t}(x | y, \sigma, \alpha, \beta)$ and one-step density $p_{t-\tau,t}(x | y, g_2(\sigma, \alpha, \beta))$, for the weighting function $w(x) = \beta + \sin(\alpha x)$. Parameter values are $\sigma=1$, $\alpha=1$, $\beta=2$. The function $v(x,y)$ varies roughly between 0.72 and 1.31. Panel b): Numerical calculation of $\delta(\sigma):=\max_{x,y}|v(x,y;\sigma)-1|$ for the weighting function $w(x) = \beta + \sin(\alpha x)$ for varying values of $\alpha$ and $\beta$. The parameter $\alpha$, which determines the wavelength of the sine, shifts the curve $\delta(\sigma)$ and varies the skewing, while retaining the height of the maximum. The parameter $\beta$ in contrast changes height of the maximum, which is the magnitude $\delta$ of the approximate robustness}
\label{f;fig3}
\end{figure}

\begin{figure}
\includegraphics{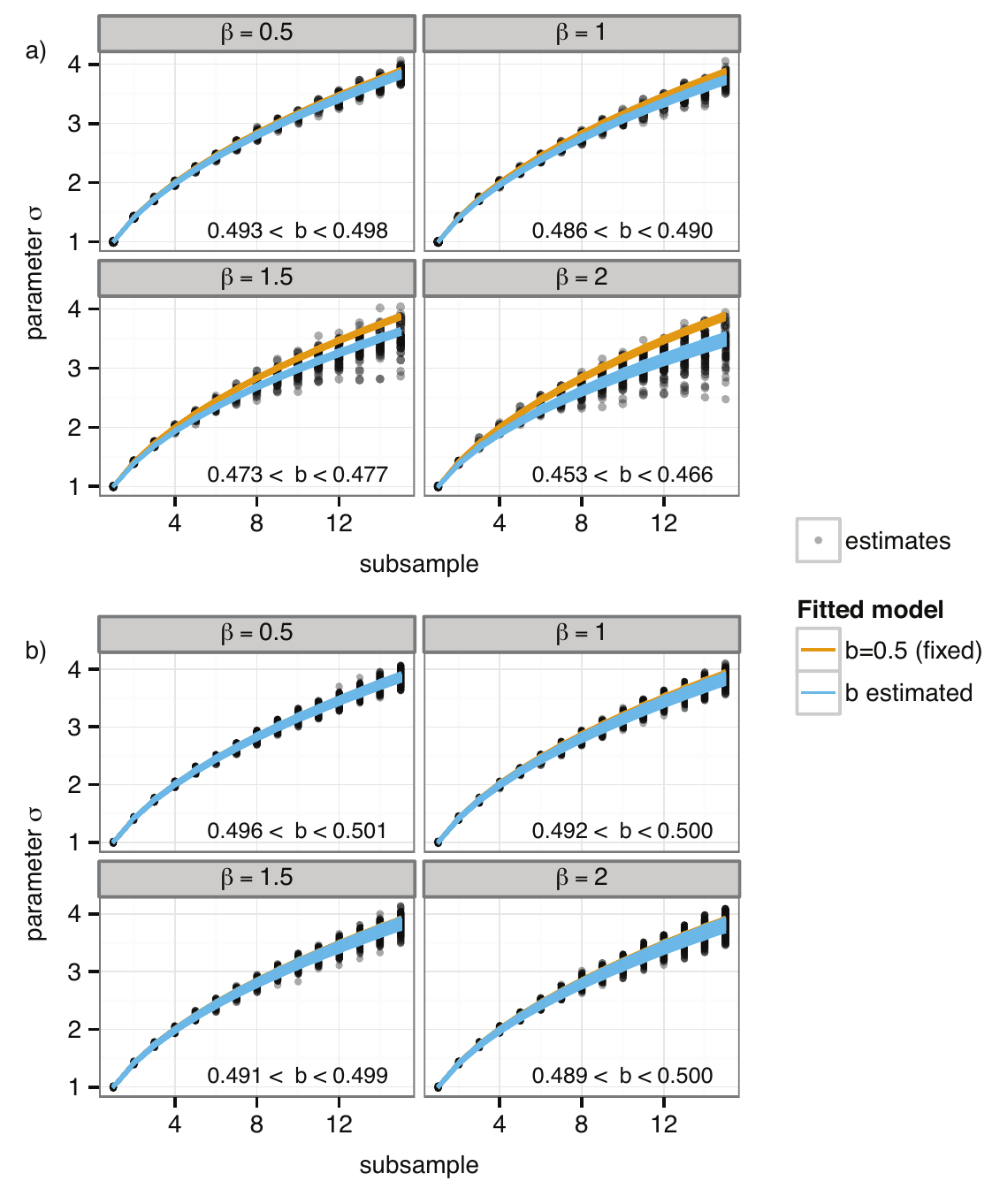}
\caption{Values of $\sigma$ against increasing subsampling amount $n$. Estimates $\hat{\sigma}$ (gray dots) were fitted with a power-relationship, stratified by trajectories, and separately for several combinations of true parameter values ($\sigma$, $\beta$, and $\alpha$ for the model with logistic RSF). The power $b$ was either fixed at 0.5 (ideal relationship; upper orange lines) or flexible and estimated (lower blue lines). The noted range of $b$ refers to variation for different parameter combinations. Estimates and predictions are standardized by the corresponding true value. Panel a): Model with exponential RSF. With increasing value of $\beta$, estimates $\hat{\sigma}$ tended to increase less with subsampling compared to the ideal relationship. Panel b): Model with logistic RSF. The fitted power-relationship was very close to the ideal relationship, such that lines indicating the ideal relationship are overlaid by lines showing the fitted relationship}
\label{f;fig4}
\end{figure}

\begin{figure}
\includegraphics{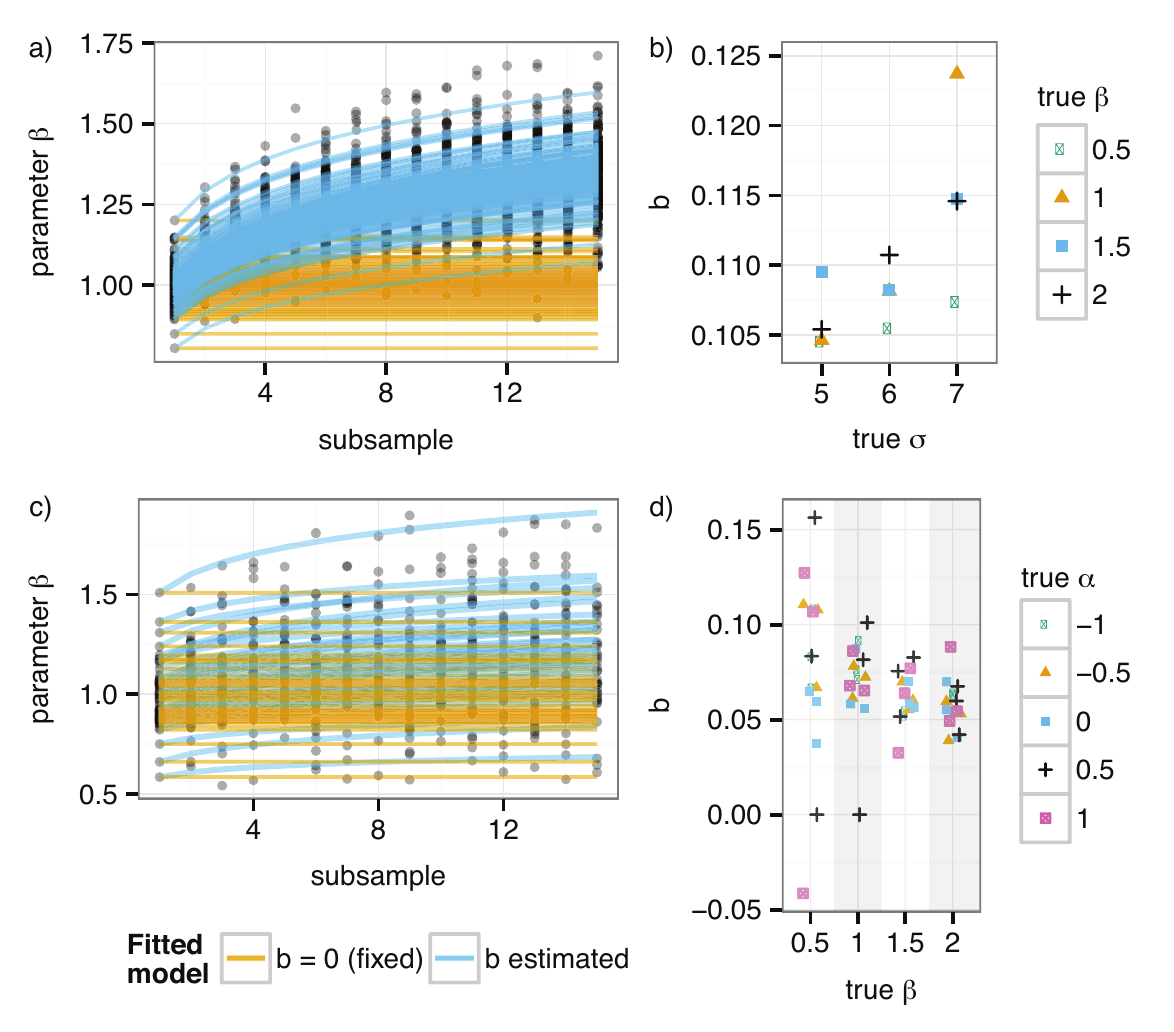}
\caption{Simulation results for the resource selection parameter $\beta$ for the model with exponential RSF (panels a,b) and logistic RSF (panels c,d). Panels a) and c): Estimates $\hat{\beta}$ (gray dots) for increasing subsampling amount $n$ were fitted with a power-relationship, stratified by trajectories, and separately for several combinations of true parameter values ($\sigma$, $\beta$, and $\alpha$ for the model with logistic RSF). The power $b$ was either fixed at zero, representing the assumption that resource-selection parameters do not change with changing temporal resolution (ideal relationship; straight orange lines), or flexible and estimated (curved blue lines). Estimates and predictions are standardized by the corresponding true value. In panel c), only estimates and predictions for $\alpha=0$, $\beta=1$ are shown. Panel b): For the exponential RSF, the estimated power $b$ was always above 0.1 and tended to increase with $\sigma$. Panel d): For the logistic RSF, the estimated power $b$ was mainly below 0.1 and tended to decrease and concentrate more for increasing $\beta$}
\label{f;fig5}
\end{figure}

\begin{figure}
\includegraphics{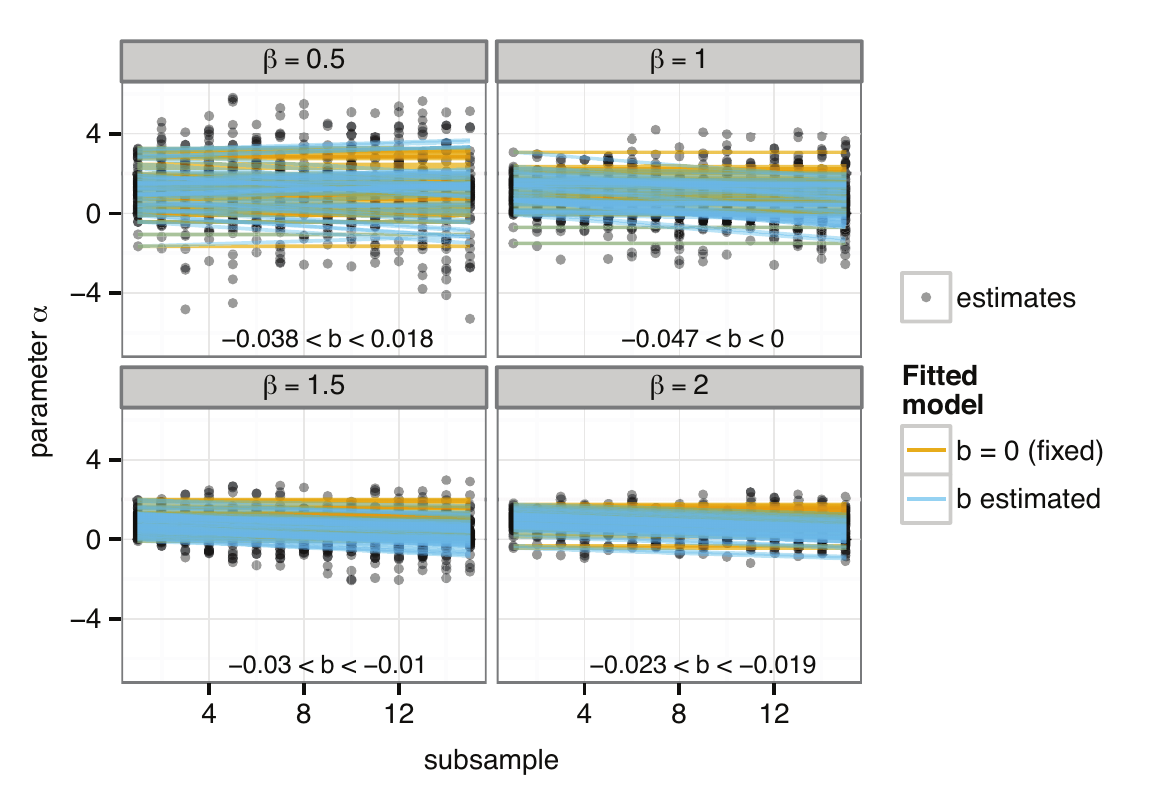}
\caption{For the model with logistic RSF, values of $\alpha$ against increasing subsampling amount $n$. Estimates were fitted with a linear relationship, stratified by trajectories, and separately for several combinations of true parameter values ($\sigma$, $\beta$, and $\alpha$ for the model with logistic RSF). The slope $b$ was either fixed at zero, representing the assumption that resource-selection parameters do not change with changing temporal resolution (ideal relationship; straight orange lines), or flexible and estimated (blue lines). Estimates and predictions are standardized by the corresponding true value and only shown for $\alpha=0.5$. The noted range of $b$ refers to variation for different parameter combinations}
\label{f;fig6}
\end{figure}

\begin{figure}
\includegraphics{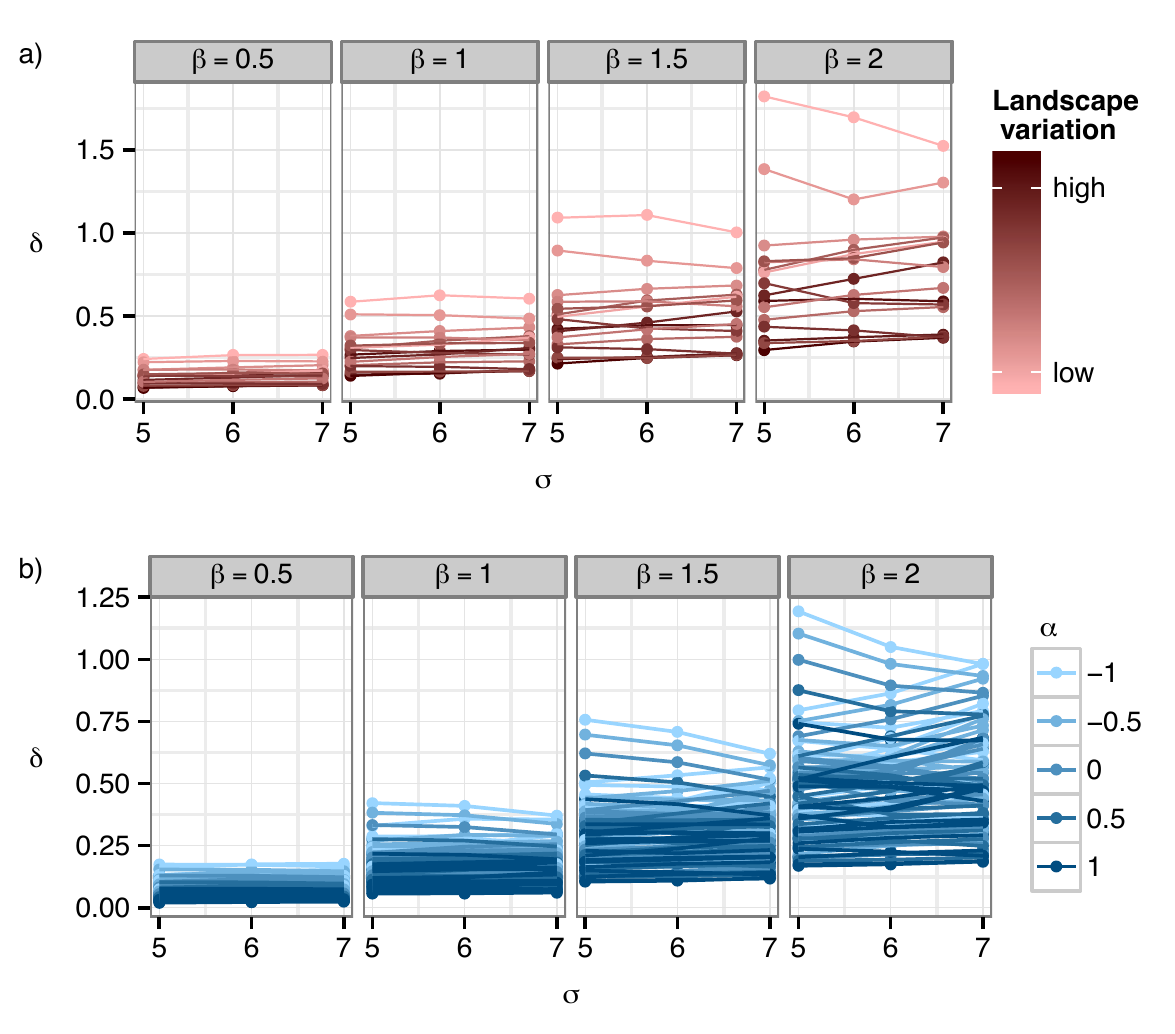}
\caption{Magnitudes of approximate robustness for the study case models with resource selection. The plots depict $\delta$ for varying values of $\sigma$ and selection parameter $\beta$ (dots). Lines join values for the same landscape $i$, $1\leq i \leq 16$. Panel a): Magnitudes for the model exponential RSF. Values of $\delta$ tend to be lower for landscapes with less variation $\text{Var}(r(x))$. Panel b): Magnitudes for the model with logistic RSF. Values of $\delta$ tend to be lower for higher values of the additional intercept parameter $\alpha$}
\label{f;fig7}
\end{figure}

\clearpage


\begin{acknowledgements}
UES was supported by a scholarship from iCORE, now part of Alberta Innovates-Technology Futures and funding from the University of Alberta.  MAL gratefully acknowledges Natural Sciences and Engineering Research Council Discovery and Accelerator grants, a Canada Research Chair and a Killam Research Fellowship.
\end{acknowledgements}


\bibliographystyle{spbasic}      
\bibliography{bib}

\begin{thebibliography}{61}
\providecommand{\natexlab}[1]{#1}
\providecommand{\url}[1]{{#1}}
\providecommand{\urlprefix}{URL }
\expandafter\ifx\csname urlstyle\endcsname\relax
  \providecommand{\doi}[1]{DOI~\discretionary{}{}{}#1}\else
  \providecommand{\doi}{DOI~\discretionary{}{}{}\begingroup
  \urlstyle{rm}\Url}\fi
\providecommand{\eprint}[2][]{\url{#2}}

\bibitem[{Aarts et~al(2011)Aarts, Fieberg, and Matthiopoulos}]{Aarts:2011fz}
Aarts G, Fieberg J, Matthiopoulos J (2011) {Comparative interpretation of
  count, presence-absence and point methods for species distribution models}.
  Methods Ecol Evol 3(1):177--187

\bibitem[{Avgar et~al(2013)Avgar, Deardon, and Fryxell}]{Avgar:2013hz}
Avgar T, Deardon R, Fryxell JM (2013) {An empirically parameterized individual
  based model of animal movement, perception, and memory}. Ecol Modell
  251:158--172

\bibitem[{Benhamou(2006)}]{Benhamou:2006kq}
Benhamou S (2006) {Detecting an orientation component in animal paths when the
  preferred direction is individual-dependent}. Ecology 87(2):518--528

\bibitem[{Benhamou(2013)}]{Benhamou:2013de}
Benhamou S (2013) {Of scales and stationarity in animal movements}. Ecol Lett
  17(3):261--272

\bibitem[{Benhamou et~al(2011)Benhamou, Sudre, Bourjea, Ciccione, De~Santis,
  and Luschi}]{Benhamou:2011jw}
Benhamou S, Sudre J, Bourjea J, Ciccione S, De~Santis A, Luschi P (2011) {The
  Role of Geomagnetic Cues in Green Turtle Open Sea Navigation}. PLoS ONE
  6(10):e26,672

\bibitem[{B{\"o}rger et~al(2008)B{\"o}rger, Dalziel, and
  Fryxell}]{Borger:2008gc}
B{\"o}rger L, Dalziel BD, Fryxell JM (2008) {Are there general mechanisms of
  animal home range behaviour? A review and prospects for future research}.
  Ecol Lett 11(6):637--650

\bibitem[{Boyce et~al(2002)Boyce, Vernier, Nielsen, and
  Schmiegelow}]{Boyce:2002vz}
Boyce M, Vernier P, Nielsen S, Schmiegelow F (2002) {Evaluating resource
  selection functions}. Ecol Modell 157(2):281--300

\bibitem[{Breed et~al(2011)Breed, Costa, Goebel, and Robinson}]{Breed:2011ux}
Breed GA, Costa DP, Goebel ME, Robinson PW (2011) {Electronic tracking tag
  programming is critical to data collection for behavioral time-series
  analysis}. Ecosphere 2(1):art10

\bibitem[{Codling and Hill(2005)}]{Codling:2005uo}
Codling E, Hill NA (2005) {Sampling rate effects on measurements of correlated
  and biased random walks.} J Theor Biol 233(4):573

\bibitem[{Codling et~al(2008)Codling, Plank, and Benhamou}]{Codling:2008js}
Codling EA, Plank MJ, Benhamou S (2008) {Random walk models in biology}. J R
  Soc Interface 5(25):813--834

\bibitem[{Colchero et~al(2010)Colchero, Conde, Manterola, Ch{\'a}vez, Rivera,
  and Ceballos}]{Colchero:2010cb}
Colchero F, Conde DA, Manterola C, Ch{\'a}vez C, Rivera A, Ceballos G (2010)
  {Jaguars on the move: modeling movement to mitigate fragmentation from road
  expansion in the Mayan Forest}. Anim Conserv 14(2):158--166

\bibitem[{C{\^o}rtes and Uriarte(2013)}]{Cortes:2013um}
C{\^o}rtes MC, Uriarte M (2013) {Integrating frugivory and animal movement: a
  review of the evidence and implications for scaling seed dispersal}. Biol Rev
  88(2):255--272

\bibitem[{Costa et~al(2012)Costa, Breed, and Robinson}]{Costa:2012ba}
Costa DP, Breed GA, Robinson PW (2012) {New Insights into Pelagic Migrations:
  Implications for Ecology and Conservation}. Annu Rev Ecol Evol Syst
  43(1):73--96

\bibitem[{Courbin et~al(2013)Courbin, Fortin, Dussault, Fargeot, and
  Courtois}]{Courbin:2013hu}
Courbin N, Fortin D, Dussault C, Fargeot V, Courtois R (2013) {Multi-trophic
  resource selection function enlightens the behavioural game between wolves
  and their prey}. J Anim Ecol 82(5):1062--1071

\bibitem[{Fleming et~al(2014)Fleming, Calabrese, Mueller, Olson, Leimgruber,
  and Fagan}]{Fleming:2014te}
Fleming CH, Calabrese JM, Mueller T, Olson KA, Leimgruber P, Fagan WF (2014)
  {From Fine-Scale Foraging to Home Ranges: A Semivariance Approach to
  Identifying Movement Modes across Spatiotemporal Scales}. Am Nat
  183(5):E154--E167

\bibitem[{Forester et~al(2009)Forester, Im, and Rathouz}]{Forester:2009wa}
Forester JD, Im H, Rathouz P (2009) {Accounting for animal movement in
  estimation of resource selection functions: sampling and data analysis}.
  Ecology 90(12):3554--3565

\bibitem[{Fortin et~al(2005)Fortin, Beyer, Boyce, Smith, Duchesne, and
  Mao}]{Fortin:2005uk}
Fortin D, Beyer H, Boyce M, Smith D, Duchesne T, Mao J (2005) {Wolves influence
  elk movements: behavior shapes a trophic cascade in Yellowstone National
  Park}. Ecology 86(5):1320--1330

\bibitem[{Frair et~al(2010)Frair, Fieberg, Hebblewhite, Cagnacci, DeCesare, and
  Pedrotti}]{Frair:2010vw}
Frair JL, Fieberg J, Hebblewhite M, Cagnacci F, DeCesare NJ, Pedrotti L (2010)
  {Resolving issues of imprecise and habitat-biased locations in ecological
  analyses using GPS telemetry data}. Philos Trans R Soc B 365(1550):2187--2200

\bibitem[{Giuggioli and Kenkre(2014)}]{Giuggioli:2014tc}
Giuggioli L, Kenkre VM (2014) {Consequences of animal interactions on their
  dynamics: emergence of home ranges and territoriality}. Mov Ecol 2:2--20

\bibitem[{Hampel(1971)}]{Hampel:1971gy}
Hampel FR (1971) {A General Qualitative Definition of Robustness}. Ann Math
  Statist 42(6):1887--1896

\bibitem[{Hampel(1986)}]{Hampel:1986vi}
Hampel FR (1986) {Robust Statistics: The approach based on influence
  functions}. Wiley, New York

\bibitem[{Haran(2011)}]{Haran:2011vv}
Haran M (2011) {Gaussian random field models for spatial data}. In: Markov
  chain Monte Carlo Handbook Eds Brooks, SP, Gelman, AE Jones, GL and Meng, XL,
  Chapman and Hall/CRC, pp 449--478

\bibitem[{Hebblewhite and Merrill(2008)}]{Hebblewhite:2008tg}
Hebblewhite M, Merrill E (2008) {Modelling wildlife--human relationships for
  social species with mixed‐effects resource selection models}. J Appl Ecol
  45(3):834--844

\bibitem[{Huber and Ronchetti(2009)}]{Huber:2009vi}
Huber PJ, Ronchetti EM (2009) {Robust Statistics}, 2nd edn. Wiley Series in
  Probability and Statistics, John Wiley {\&} Sons, Inc., Hoboken, N.J.

\bibitem[{Ito et~al(2013)Ito, Lhagvasuren, Tsunekawa, Shinoda, Takatsuki,
  Buuveibaatar, and Chimeddorj}]{Ito:2013hg}
Ito TY, Lhagvasuren B, Tsunekawa A, Shinoda M, Takatsuki S, Buuveibaatar B,
  Chimeddorj B (2013) {Fragmentation of the Habitat of Wild Ungulates by
  Anthropogenic Barriers in Mongolia}. PLoS ONE 8(2):e56,995

\bibitem[{Jerde and Visscher(2005)}]{Jerde:2005fp}
Jerde CL, Visscher DR (2005) {GPS measurement error influences on movement
  model parameterization}. Ecol Appl 15(3):806--810

\bibitem[{Johnson et~al(2002)Johnson, Parker, Heard, and
  Gillingham}]{Johnson:2002hs}
Johnson CJ, Parker KL, Heard DC, Gillingham MP (2002) {Movement parameters of
  ungulates and scale-specific responses to the environment}. J Anim Ecol
  71(2):225--235

\bibitem[{Kareiva and Shigesada(1983)}]{Kareiva:1983vf}
Kareiva PM, Shigesada N (1983) {Analyzing insect movement as a correlated
  random walk}. Oecologia 56(2):234--238

\bibitem[{Langrock et~al(2013)Langrock, King, Matthiopoulos, Thomas, Fortin,
  and Morales}]{Langrock:2012ek}
Langrock R, King R, Matthiopoulos J, Thomas L, Fortin D, Morales JM (2013)
  {Flexible and practical modeling of animal telemetry data: hidden Markov
  models and extensions}. Ecology 93(11):2336--2342

\bibitem[{Latham et~al(2011)Latham, Latham, Boyce, and Boutin}]{Latham:2011uq}
Latham ADM, Latham MC, Boyce M, Boutin S (2011) {Movement responses by wolves
  to industrial linear features and their effect on woodland caribou in
  northeastern Alberta}. Ecol Appl 21(8):2854--2865

\bibitem[{Lele and Keim(2006)}]{Lele:2006wc}
Lele SR, Keim JL (2006) {Weighted distributions and estimation of resource
  selection probability functions}. Ecology 87(12):3021--3028

\bibitem[{Lele et~al(2013)Lele, Merrill, Keim, and Boyce}]{Lele:2013hw}
Lele SR, Merrill EH, Keim J, Boyce MS (2013) {Selection, use, choice and
  occupancy: clarifying concepts in resource selection studies}. J Anim Ecol
  82(6):1183--1191

\bibitem[{Manly et~al(2002)Manly, McDonald, Thomas, McDonald, and
  Erickson}]{Manly:2002un}
Manly BF, McDonald LL, Thomas DL, McDonald TL, Erickson WP (2002) {Resource
  selection by animals: statical design and analysis for field studies}, 2nd
  edn. Kluwer Academic Publishers, Dordrecht

\bibitem[{Masden et~al(2012)Masden, Reeve, Desholm, Fox, Furness, and
  Haydon}]{Masden:2012bp}
Masden EA, Reeve R, Desholm M, Fox AD, Furness RW, Haydon DT (2012) {Assessing
  the impact of marine wind farms on birds through movement modelling}. J R Soc
  Interface 9(74):2120--2130

\bibitem[{McClintock et~al(2012)McClintock, King, Thomas, Matthiopoulos,
  McConnell, and Morales}]{McClintock:2012bw}
McClintock BT, King R, Thomas L, Matthiopoulos J, McConnell BJ, Morales JM
  (2012) {A general discrete-time modeling framework for animal movement using
  multistate random walks}. Ecol Monogr 82(3):335--349

\bibitem[{McClintock et~al(2014)McClintock, Johnson, Hooten, Ver~Hoef, and
  Morales}]{McClintock:2014in}
McClintock BT, Johnson DS, Hooten MB, Ver~Hoef JM, Morales JM (2014) {When to
  be discrete: the importance of time formulation in understanding animal
  movement}. Mov Ecol 2(1):334

\bibitem[{McPhee et~al(2012)McPhee, Webb, and Merrill}]{McPhee:2012ii}
McPhee HM, Webb NF, Merrill EH (2012) {Hierarchical predation: wolf
  (\textit{Canis lupus}) selection along hunt paths and at kill sites}. Can J
  Zool 90(5):555--563

\bibitem[{Mills et~al(2006)Mills, Patterson, and Murray}]{Mills:2006tn}
Mills KJ, Patterson BR, Murray DL (2006) {Effects of variable sampling
  frequencies on GPS transmitter efficiency and estimated wolf home range size
  and movement distance}. Wildl Soc Bull 34(5):1463--1469

\bibitem[{Moorcroft and Barnett(2008)}]{Moorcroft:2008wa}
Moorcroft PR, Barnett A (2008) {Mechanistic Home Range Models and Resource
  Selection Analysis: A Reconciliation and Unification}. Ecology
  89(4):1112--1119

\bibitem[{Mueller et~al(2014)Mueller, Lenz, Caprano, Fiedler, and
  B{\"o}hning-Gaese}]{Mueller:2014hq}
Mueller T, Lenz J, Caprano T, Fiedler W, B{\"o}hning-Gaese K (2014) {Large
  frugivorous birds facilitate functional connectivity of fragmented
  landscapes}. J Appl Ecol 51(3):684--692

\bibitem[{P{\'e}pin et~al(2004)P{\'e}pin, Adrados, Mann, and
  Janeau}]{Pepin:2004bh}
P{\'e}pin D, Adrados C, Mann C, Janeau G (2004) {Assessing real daily distance
  traveled by ungulates using differential GPS locations}. J Mammal
  85(4):774--780

\bibitem[{Postlethwaite and Dennis(2013)}]{Postlethwaite:2013ki}
Postlethwaite CM, Dennis TE (2013) {Effects of Temporal Resolution on an
  Inferential Model of Animal Movement}. PLoS ONE 8(5):e57,640

\bibitem[{Potts and Lewis(2014)}]{Potts:2014dr}
Potts JR, Lewis MA (2014) {How do animal territories form and change? Lessons
  from 20 years of mechanistic modelling}. Proc R Soc B 281(1784):20140,231

\bibitem[{Potts et~al(2014)Potts, Bastille-Rousseau, Murray, Schaefer, and
  Lewis}]{Potts:2014wt}
Potts JR, Bastille-Rousseau G, Murray DL, Schaefer JA, Lewis MA (2014)
  {Predicting local and non-local effects of resources on animal space use
  using a mechanistic step selection model}. Methods Ecol Evol 5(3):253--262

\bibitem[{Rhodes et~al(2005)Rhodes, McAlpine, Lunney, and
  Possingham}]{Rhodes:2005uu}
Rhodes JR, McAlpine CA, Lunney D, Possingham HP (2005) {A spatially explicit
  habitat selection model incorporating home range behavior}. Ecology
  86(5):1199--1205

\bibitem[{Robert and Casella(2000)}]{robert2000monte}
Robert CP, Casella G (2000) {Monte Carlo statistical methods}, corrected 2.
  print. edn. Springer texts in statistics, Springer, New York

\bibitem[{Robinson et~al(2009)Robinson, Bowlin, Bisson, Shamoun-Baranes,
  Thorup, Diehl, Kunz, Mabey, and Winkler}]{Robinson:2009br}
Robinson WD, Bowlin MS, Bisson I, Shamoun-Baranes J, Thorup K, Diehl RH, Kunz
  TH, Mabey S, Winkler DW (2009) {Integrating concepts and technologies to
  advance the study of bird migration}. Frontiers in Ecology and the
  Environment 8(7):354--361

\bibitem[{Rosser et~al(2013)Rosser, Fletcher, Maini, and Baker}]{Rosser:2013ia}
Rosser G, Fletcher AG, Maini PK, Baker RE (2013) {The effect of sampling rate
  on observed statistics in a correlated random walk}. J R Soc Interface
  10(85):20130,273

\bibitem[{Rowcliffe et~al(2012)Rowcliffe, Carbone, Kays, Kranstauber, and
  Jansen}]{Rowcliffe:2012hp}
Rowcliffe MJ, Carbone C, Kays R, Kranstauber B, Jansen PA (2012) {Bias in
  estimating animal travel distance: the effect of sampling frequency}. Methods
  Ecol Evol 3(4):653--662

\bibitem[{Ryan et~al(2004)Ryan, Petersen, Peters, and Gremillet}]{Ryan:2004et}
Ryan PG, Petersen SL, Peters G, Gremillet D (2004) {GPS tracking a marine
  predator: the effects of precision, resolution and sampling rate on foraging
  tracks of African Penguins}. Mar Biol 145(2)

\bibitem[{Sawyer et~al(2009)Sawyer, Kauffman, Nielson, and
  Horne}]{Sawyer:2009uq}
Sawyer H, Kauffman M, Nielson R, Horne J (2009) {Identifying and prioritizing
  ungulate migration routes for landscape-level conservation}. Ecol Appl
  19(8):2016--2025

\bibitem[{Schick et~al(2008)Schick, Loarie, Colchero, Best, Boustany, Conde,
  Halpin, Joppa, McClellan, and Clark}]{Schick:2008dn}
Schick RS, Loarie SR, Colchero F, Best BD, Boustany A, Conde DA, Halpin PN,
  Joppa LN, McClellan CM, Clark JS (2008) {Understanding movement data and
  movement processes: current and emerging directions}. Ecol Lett
  11(12):1338--1350

\bibitem[{Schlather et~al(2013)Schlather, Menck, Singleton, Pfaff, and {R Core
  team}}]{RandomFields:2013}
Schlather M, Menck P, Singleton R, Pfaff B, {R Core team} (2013) {RandomFields:
  Simulation and Analysis of Random Fields}

\bibitem[{Smouse et~al(2010)Smouse, Focardi, Moorcroft, Kie, Forester, and
  Morales}]{Smouse:2010ku}
Smouse PE, Focardi S, Moorcroft PR, Kie JG, Forester JD, Morales JM (2010)
  {Stochastic modelling of animal movement}. Philos Trans R Soc B
  365(1550):2201--2211

\bibitem[{Squires et~al(2013)Squires, DeCesare, Olson, Kolbe, Hebblewhite, and
  Parks}]{Squires:2013hs}
Squires JR, DeCesare NJ, Olson LE, Kolbe JA, Hebblewhite M, Parks SA (2013)
  {Combining resource selection and movement behavior to predict corridors for
  Canada lynx at their southern range periphery}. Biol Conserv 157(0):187--195

\bibitem[{Tanferna et~al(2012)Tanferna, L{\'o}pez-Jim{\'e}nez, Blas, Hiraldo,
  and Sergio}]{Tanferna:2012fc}
Tanferna A, L{\'o}pez-Jim{\'e}nez L, Blas J, Hiraldo F, Sergio F (2012)
  {Different Location Sampling Frequencies by Satellite Tags Yield Different
  Estimates of Migration Performance: Pooling Data Requires a Common Protocol}.
  PLoS ONE 7(11):e49,659

\bibitem[{Tsoar et~al(2011)Tsoar, Nathan, Bartan, Vyssotski, Dell'Omo, and
  Ulanovsky}]{Tsoar:2011wz}
Tsoar A, Nathan R, Bartan Y, Vyssotski A, Dell'Omo G, Ulanovsky N (2011)
  {Large-scale navigational map in a mammal}. Proceedings of the National
  Academy of Sciences 108(37):E718--E724

\bibitem[{Turchin(1998)}]{Turchin:1998td}
Turchin P (1998) {Quantitative analysis of movement: measuring and modeling
  population redistribution in animals and plants}. Sinauer Associates,
  Sunderland, Mass.

\bibitem[{Wiens(2000)}]{Wiens:2000cq}
Wiens DP (2000) {Bias Constrained Minimax Robust Designs for Misspecified
  Regression Models}. In: Balakrishnan N, Melas VB, Ermakov S (eds) Statistics
  for Industry and Technology, Birkh{\"a}user, Boston, MA, pp 117--133

\bibitem[{Wiens and Zhou(1996)}]{Wiens:1996dm}
Wiens DP, Zhou J (1996) {Minimax regression designs for approximately linear
  models with autocorrelated errors}. J Statist Plann Inference 55(1):95--106

\bibitem[{Wilcox(2012)}]{Wilcox:2012wh}
Wilcox RR (2012) {Introduction to Robust Estimation and Hypothesis Testing},
  3rd edn. Academic Press, Boston

\end{thebibliography}


\appendix

\section{Proofs of results about exact robustness}
\label{s;proofs exact}

\begin{proof}[Theorem~\ref{t;linear w}]
First, note that for any standard deviation of the kernel, $\sigma$, the integral $\int_{\mathbb{R}} k_{\sigma}(y;x)w(y)\,\D y$ reduces to the weighting function evaluated at the kernel's mean,
\begin{multline}\label{e;linear integral}
\int_{\mathbb{R}} k_{\sigma}(y;x)w(y)\,\D y = \int_{\mathbb{R}} k_{\sigma}(y;x)(ay+b)\,\D y = \int_{\mathbb{R}} k_{\sigma}(y;x)(a(y-x+x)+b)\,\D y \\
 = (ax+b)\int_{\mathbb{R}} k_{\sigma}(y;x)\,\D y + a\int_{\mathbb{R}} k_{\sigma}(y;x)(y-x)\,\D y = ax+b = w(x),
\end{multline}
because $k_{\sigma}(\cdot| y)$ is a Gaussian density integrating to 1 and with vanishing first central moment. If we consider $w$ as a linear transformation of a Normally distributed random variable with mean $x$, then equation~\eqref{e;linear integral} reflects a special case of Jensen's inequality, in which equality holds.

We now show robustness of degree $n$ with parameter transformation $g_n(\sigma,a,b) = (\sqrt{n}\sigma,a,b)$ by induction. 
For $n=1$, we have the trivial transformation $g_1(\sigma, a, b) = (\sigma, a, b)$, and there is nothing to show for robustness of degree 1.

We assume that robustness or degree $n$ holds, that is we have the relationship
\begin{equation}
 p_{n}(x_n | x_0, \sigma, a, b) = p_1(x_n|x_0, \sqrt{n}\sigma, a, b).
\end{equation}
for all $x_n, x_0 \in\mathbb{R}$. For $n+1$, we use the Chapman-Kolmogorov equation and Markov property and obtain
\begin{align}
 p_{n+1}(x_{n+1}|x_0, \sigma, a, b) 
 &= \int_{\mathbb{R}^n} \prod_{k=1}^{n+1} p_1(x_k|x_{k-1}, \sigma, a, b) \, \D x_1 \dots \D x_n    \notag \\
 &= \int_{\mathbb{R}} p_1(x_{n+1}|x_n, \sigma, a, b) \left( \int_{\mathbb{R}^{n-1}} \prod_{k=1}^{n} p_1(x_k|x_{k-1}, \sigma, a, b) \, \D x_1 \dots \D x_{n-1} \right) \D x_n   \notag \\
 &= \int_{\mathbb{R}} p_1(x_{n+1}|x_n, \sigma, a, b) \, p_{n}(x_n|x_0, \sigma, a, b) \, \D x_n  \notag \\
 &= \int_{\mathbb{R}} p_1(x_{n+1}|x_n, \sigma, a, b) \, p_1(x_n|x_0, \sqrt{n}\sigma, a, b) \, \D x_n,
\end{align}
where the last step follows by induction. We can now insert the model's step probabilities and use equation~\eqref{e;linear integral} to further calculate,
\begin{align}
 p_{n+1}(x_{n+1}|x_0, \sigma, a, b)
 &= \int_{\mathbb{R}} \frac{k_{\sigma}(x_{n+1};x_n)\,w(x_{n+1})}{\int_{\mathbb{R}} k_{\sigma}(y;x_n) w(y) \, dy } \, \frac{k_{\sqrt{n}\,\sigma}(x_n;x_0)\,w(x_n)}{\int_{\mathbb{R}} k_{\sqrt{n}\,\sigma}(y;x_0) w(y)\, dy} \, \D x_n  \notag \\[8pt]
 &= \int_{\mathbb{R}} \frac{k_{\sigma}(x_{n+1};x_n)\,w(x_{n+1})}{w(x_n)} \, \frac{k_{\sqrt{n}\,\sigma}(x_n;x_0)\,w(x_n)}{w(x_0)} \, \D x_n  \notag \\[8pt]
 &= \frac{w(x_{n+1})}{w(x_0)} \int_{\mathbb{R}} k_{\sigma}(x_{n+1};x_n)\,k_{\sqrt{n}\sigma}(x_n;x_0) \, \D z.
\end{align}
Note that we have assumed that all movement steps are within the domain $\mathcal{I}$, where the weighting function is positive. Since $k_{\sigma}(x_{n+1};x_n) = k_{\sigma}(x_{n+1}-x_n;0)$, the integral in the last expression is the convolution of two Gaussian densities with variances $\sigma^2$ and $n\sigma^2$ and with means 0 and $x_0$, respectively. Because of the linearity of Gaussian random variables, this is again a Gaussian density with mean $x_0$ and variance $(n+1)\sigma^2$. Because equation~\eqref{e;linear integral} holds for the kernel with any standard deviation, we can rewrite the denominator as $w(x_0) = \int_{\mathbb{R}} k_{\sqrt{n+1}\,\sigma}(y;x_0) w(y) \, dy$. Thus,
\begin{equation}
 p_{n+1}(x_{n+1}|x_0, \sigma, a, b) 
 = \frac{k_{\sqrt{n+1}\,\sigma}(x_{n+1};x_0) w(x_{n+1})}{\int_{\mathbb{R}} k_{\sqrt{n+1}\,\sigma}(y;x_0) w(y) \, dy}
 = p_{1}(x_{n+1}|x_0,\sqrt{n+1}\,\sigma, a, b).
\end{equation}
\qed
\end{proof}


\begin{proof}[Theorem~\ref{t;exp w}]
We proceed analogously to the previous proof. The integral of weighting function and kernel with arbitrary standard deviation $\sigma$ and mean $x$ is here given by
\begin{align}
 \int_{\mathbb{R}} k_{\sigma}(y;x) \, w(y)\,\D y 
 &= \int_{\mathbb{R}} k_{\sigma}(y;x)\,C e^{ay+b}\,\D y   \notag \\
 &= \frac{C}{\sqrt{2\pi}\sigma} \int_{\mathbb{R}} \exp \left( -\frac{(y-x)^2}{2\sigma^2} + ay + b\right) \D y.   \notag
\intertext{By completing the square and using substitution $u=\frac{1}{\sqrt{2} \sigma} (y-x-a\sigma^2)$ we obtain}
 \int_{\mathbb{R}} k_{\sigma}(y;x) \, w(y)\,\D y 
  &= \frac{C}{\sqrt{2\pi}\sigma} e^{ \frac{a^2\sigma^2}{2} +ax +b } \int_{\mathbb{R}} \exp \left( - \left( \frac{y-x-a\sigma^2}{\sqrt{2}\sigma} \right)^2 \right) \D y   \notag \\
  &= \frac{C}{\sqrt{2\pi}\sigma} e^{ \frac{a^2\sigma^2}{2} +ax +b } \int_{\mathbb{R}} \exp \left( - u^2 \right) \sqrt{2} \sigma \, \D u.  \notag
\intertext{The final integral reduces to $\sqrt{2\pi}\sigma$, and therefore,}
 \int_{\mathbb{R}} k_{\sigma}(y;x) \, w(y)\,\D y &= C\, e^{\frac{a^2\sigma^2}{2} + ax + b}. \label{e;exponential integral}
\end{align}

Again, we prove robustness of degree $n$ by induction, using parameter transformation $g_n(\sigma, C, a, b) = (\sqrt{n}\sigma, C,a,b)$. In the induction step, we obtain, with help of equation~\eqref{e;exponential integral},
\begin{align}
 p_{n+1}(x_{n+1}|x_0, \sigma, a, b)
 &= \int_{\mathbb{R}} \frac{k_{\sigma}(x_{n+1};x_n)\,Ce^{ax_{n+1}+b}}{\int_{\mathbb{R}} k_{\sigma}(y;x_n) Ce^{ay+b} \, dy } \, \frac{k_{\sqrt{n}\,\sigma}(x_n;x_0)\,Ce^{ax_n+b}}{\int_{\mathbb{R}} k_{\sqrt{n}\,\sigma}(y;x_0) Ce^{ay+b}\, dy} \, \D x_n  \notag  \\[8pt]
 &= \int_{\mathbb{R}} \frac{k_{\sigma}(x_{n+1};x_n)\,Ce^{ax_{n+1}+b}}{Ce^{\frac{a^2\sigma^2}{2} + ax_n + b}} \, \frac{k_{\sqrt{n}\,\sigma}(x_n;x_0)\,Ce^{ax_n+b}}{Ce^{\frac{na^2\sigma^2}{2} + ax_0 + b}} \, \D x_n  \notag \\[8pt]
 &= \frac{e^{x_{n+1}}}{e^{\frac{(n+1)a^2\sigma^2}{2}+ax_0}} \int_{\mathbb{R}} k_{\sigma}(x_{n+1};x_n)\,k_{\sqrt{n}\,\sigma}(x_n;x_0) \, \D z  \notag \\[8pt]
 &= \frac{e^{x_{n+1}}}{e^{\frac{(n+1)a^2\sigma^2}{2}+ax_0}} \, k_{\sqrt{n+1}\,\sigma}(x_{n+1}; x_0). \label{e;t2 lhs}   \notag \\[8pt]
 &= \frac{k_{\sqrt{n+1}\,\sigma}(x_{n+1}; x_0) \, C e^{ax_{n+1}+b}}{\int_{\mathbb{R}} k_{\sqrt{n+1}\,\sigma}(y; x_0) \, C e^{ay+b} \,\D y}  \notag \\[8pt]
 &= p_{1}(x_{n+1}|x_0, \sqrt{n+1}\,\sigma, a, b)
\end{align}
\qed
\end{proof}


\section{Proof of result about asymptotic robustness}
\label{s;proof asymptotic}

To highlight the main steps necessary to prove Theorem~\ref{t;asymptotic robustness}, we establish a series of intermediate results. As a first step, we show that the 2-step transition density can be broken up into a product of the form~\eqref{e;definition asymptotic} in Definition~\ref{t;definition asymptotic}.

\begin{proposition}\label{t;p times v}
The 2-step transition density of model with transitions~\eqref{e;transition density} can be written as
\begin{equation}\label{e;desired times v}
 p_2(x_t|x_{t-2\tau}, \sigma, \boldsymbol{\theta}) = p_1(x_t|x_{t-2\tau},\sqrt{2}\sigma,\boldsymbol{\theta}) \cdot v(x_t, x_{t-2\tau}; \tau),
\end{equation}
where the function $v$ is given by 
\begin{equation}\label{e;function v}
v(x_t, x_{t-2\tau}; \tau)
= \frac{\int_{\mathbb{R}} k_{\sqrt{2}\sigma}(y;x) w_{\boldsymbol{\theta}}(y)\, \D y}{\int_{\mathbb{R}} k_{\sigma}(y;x) w_{\boldsymbol{\theta}}(y)\, \D y}  
\int_{\mathbb{R}} k_{\frac{\sigma}{\sqrt{2}}}\Bigl(z; \frac{1}{2}(x_t+x_{t-2\tau})\Bigr) \, \frac{w_{\boldsymbol{\theta}}(z)}{\int_{\mathbb{R}} k_{\sigma}(y;z) w_{\boldsymbol{\theta}}(y) \, \D y} \, \D z. 
\end{equation}
\end{proposition}
Note that $v$ depends on $\tau$ through $\sigma$. For later convenience, we define
\begin{align}
 Q(x; \tau)&:= \frac{\int_{\mathbb{R}} k_{\sqrt{2}\sigma}(y;x) w_{\boldsymbol{\theta}}(y)\, \D y}{\int_{\mathbb{R}} k_{\sigma}(y;x) w_{\boldsymbol{\theta}}(y)\, \D y} \label{e;definition Q} \\
 I(x_1,x_2; \tau) &:= \int_{\mathbb{R}} k_{\frac{\sigma}{\sqrt{2}}}\Bigl(z; \frac{1}{2}(x_1+x_2)\Bigr) \, \frac{w_{\boldsymbol{\theta}}(z)}{\int_{\mathbb{R}} k_{\sigma}(y;z) w_{\boldsymbol{\theta}}(y) \, \D y} \, \D z. \label{e;definition I}
\end{align}

\begin{proof}
The proposition can be shown with a straightforward calculation. The 2-step transition density is given by
\begin{align}
p_2(x_t|x_{t-2\tau}&,\sigma,\boldsymbol{\theta}) \\
 &= \int_{\mathbb{R}} \frac{k_{\sigma}(x_t;z)\,w_{\boldsymbol{\theta}}(x_t)}{\int_{\mathbb{R}} k_{\sigma}(y;z) w_{\boldsymbol{\theta}}(y) \, \D y } \, \frac{k_{\sigma}(z;x_{t-2\tau})\,w_{\boldsymbol{\theta}}(z)}{\int_{\mathbb{R}} k_{\sigma}(y;x_{t-2\tau}) w_{\boldsymbol{\theta}}(y)\, \D y} \, \D z \\
 &= \frac{w_{\boldsymbol{\theta}}(x_t)}{\int_{\mathbb{R}} k_{\sigma}(y;x_{t-2\tau}) w_{\boldsymbol{\theta}}(y)\, \D y} \int_{\mathbb{R}} k_{\sigma}(x_t;z)\,k_{\sigma}(z;x_{t-2\tau}) \, \frac{w_{\boldsymbol{\theta}}(z)}{\int_{\mathbb{R}} k_{\sigma}(y;z) w_{\boldsymbol{\theta}}(y) \, \D y} \, \D z.
\end{align}
Tthe product of the two Gaussian densities in the integrand can be transformed as follows
\begin{equation}
k_{\sigma}(x_t;z)\,k_{\sigma}(z;x_{t-2\tau})
 = k_{\sqrt{2}\sigma}(x_t;x_{t-2\tau}) \, k_{\frac{\sigma}{\sqrt{2}}}\Bigl(z; \frac{1}{2}(x_t+x_{t-2\tau})\Bigr).
\end{equation}
The two-step density therefore becomes
\begin{multline}
p_2(x_t|x_{t-2\tau},\sigma,\boldsymbol{\theta}) \\
 = \frac{k_{\sqrt{2}\sigma}(x_t;x_{t-2\tau}) \,w_{\boldsymbol{\theta}}(x_t)}{\int_{\mathbb{R}} k_{\sigma}(y;x_{t-2\tau}) w_{\boldsymbol{\theta}}(y)\, \D y} \int_{\mathbb{R}} k_{\frac{\sigma}{\sqrt{2}}}\Bigl(z; \frac{1}{2}(x_t+x_{t-2\tau})\Bigr) \, \frac{w_{\boldsymbol{\theta}}(z)}{\int_{\mathbb{R}} k_{\sigma}(y;z) w_{\boldsymbol{\theta}}(y) \, \D y} \, \D z.
\end{multline}
The numerator of the first factor is the desired one-step density up to appropriate normalization. If we extend by the required normalization constant $\int_{\mathbb{R}} k_{\sqrt{2}\sigma}(y;x_{t-2\tau}) w_{\boldsymbol{\theta}}(y)\, \D y$, we obtain equations~\eqref{e;desired times v} and~\eqref{e;function v}.
\qed
\end{proof}

We are now left to show that the function $v-1$ is in the order of $\tau$ on its entire domain $\mathbb{R}^2\times\mathbb{R}^+$. In particular, this means that for any fixed $\tau^{\ast}$, the function $v(x_1, x_2; \tau^{\ast})-1$ is bounded on $\mathbb{R}^2$ via $c\tau^{\ast}$ for a constant $c$. It turns out to be helpful to analyze $v$ separately on $\mathbb{R}^2\times (0,\tau_0)$ and $\mathbb{R}^2\times [\tau_0,\infty)$ for some $\tau_0$. Because the proof is simpler for large $\tau$, we present this result first.

\begin{lemma}\label{t;v is of order tau for large tau}
Let $w$ be continuous and bounded away from zero, that is there exist $L$ and $U$ such that $0<L\leq w_{\boldsymbol{\theta}}(x) \leq U$ for all $x\in \mathbb{R}$. Let $w$ further be twice differentiable on $\mathbb{R}$ with $|w''(x)|<M$ for some $M$ and all $x\in \mathbb{R}$. For any $\tau_0 > 0$, we have $v(x_1, x_2,;\tau) -1 = \mathcal{O}(\tau)$ on $\mathbb{R}^2\times[\tau_0,\infty)$.
\end{lemma}

\begin{proof}
Let $\tau_0$ be a number away from zero and fixed. Our goal is to establish bounds on the functions $Q$ and $I$, as defined in~\eqref{e;definition Q} and~\eqref{e;definition I}, and to use these to place a bound on $v-1$. Because $w$ is twice differentiable we can apply Taylor's theorem to obtain a linear approximation for $w$ using any point $x\in \mathbb{R}$, 
\begin{equation}
 w_{\boldsymbol{\theta}}(y) = w_{\boldsymbol{\theta}}(x) + w'(x)(y-x)+R(y),
\end{equation}
where $R(y)$ is the remainder term. This leads to
\begin{equation}
 \int_{\mathbb{R}} k_{\sigma}(y;x)\, w_{\boldsymbol{\theta}}(y)\, \D y
 = w_{\boldsymbol{\theta}}(x) \smash{\int_{\mathbb{R}}} k_{\sigma}(y;x)\,\D y + w'(x) \smash{\int_{\mathbb{R}}} k_{\sigma}(y;x)\,(y-x)\,\D y + \smash{\int_{\mathbb{R}}} k_{\sigma}(y;x)\,R(y)\,\D y,
\end{equation}
where the first term on the RHS becomes $w_{\boldsymbol{\theta}}(x)$, because the kernel integrates to 1, and the integral in the second term is the first central moment of the kernel, hence vanishes. The remainder $R(y)$, using the Lagrange form, is given by $R(y) = \frac{w''(\xi)}{2}(y-x)^2$, for some $\xi$ between $x_2$ and $y$. Since the second derivative of $w$ is assumed to be globally bounded, we have $|R(y)|\leq \frac{M}{2}(y-x)^2$. We use this to place bounds on the third term, recognizing that the remaining integral \text{$\int_{\mathbb{R}} k_{\sigma}(y;x)\,(y-x)^2\,\D y$} is the second central moment of the Gaussian kernel $k_{\sigma}$, which is given by its variance $\sigma^2 = \omega^2\tau$. Therefore,
\begin{equation}\label{e;Taylor bounds on integral}
 w_{\boldsymbol{\theta}}(x) - \frac{M}{2}\omega^2\tau \leq \int_{\mathbb{R}} k_{\sigma}(y;x)\, w_{\boldsymbol{\theta}}(y)\, \D y \leq w_{\boldsymbol{\theta}}(x) + \frac{M}{2}\omega^2\tau. 
\end{equation}
In general, the lower bound can be arbitrarily close to zero, therefore we cannot simply invert this inequality to obtain an estimate on the inverse of the integral. Instead, we use the bounds on $w$ and again the fact $\int_{\mathbb{R}} k_{\sigma}(y;x)\,\D y=1$ for any $\sigma$ and any $x\in \mathbb{R}$ to establish
\begin{equation}\label{e;w bounds on integral}
 0 < L \leq \int_{\mathbb{R}} k_{\sigma}(y;x) \, w_{\boldsymbol{\theta}}(y) \,\D y \leq U,
\end{equation}
which can be inverted. 
Since inequalities~\eqref{e;Taylor bounds on integral} and~\eqref{e;w bounds on integral} hold for any $\sigma$ and any $x \in \mathbb{R}$, they allow us to place bounds on both $Q$ and $I$. For $Q$, we obtain
\begin{equation}
 \frac{1}{U}\bigl(w_{\boldsymbol{\theta}}(x)-M\omega^2\tau\bigr) \leq Q(x;\tau) \leq \frac{1}{L}\bigl( w_{\boldsymbol{\theta}}(x)+M\omega^2\tau \bigr)
\end{equation}
for all $x\in \mathbb{R}$, $\tau\in \mathbb{R}^+$. We can avoid the dependency of the bounds on $x$ by again invoking the bounds on $w$,
\begin{equation}
 \frac{1}{U}\bigl(L-M\omega^2\tau\bigr) \leq Q(x) \leq \frac{1}{L}\bigl( U+M\omega^2\tau \bigr).
\end{equation}
For the function $I$, we only make use of the bounds on $w$ and inequality~\eqref{e;w bounds on integral} and get
\begin{equation}
 0 < \frac{L}{U} \leq I(x_1, x_2;\tau) \leq \frac{U}{L}
\end{equation}
for all $x_1, x_2 \in \mathbb{R}$, $\tau\in \mathbb{R}^+$. We can now continue to calculate $v-1$. An upper bound is immediately given by
\begin{equation}
 v(x_1, x_2;\tau) -1 = Q(x_1;\tau)\,I(x_1, x_2;\tau) -1 \leq \frac{U^2-L^2}{L^2} + \frac{MU}{L^2}\omega^2\tau.
\end{equation}
With only few more additional steps, we obtain a lower bound by simply drawing upon $L \leq U$, its squared version and its inverse,
\begin{equation}
 -\big(v(x_1, x_2;\tau)-1\big) \leq \frac{U^2-L^2}{U^2} + \frac{ML}{U^2}\omega^2\tau \leq \frac{U^2-L^2}{L^2} + \frac{MU}{L^2}\omega^2\tau.
\end{equation}
Define $C:= \frac{U^2-L^2}{L^2\tau_0} + \frac{MU}{L^2}\omega^2$ for the $\tau_0$ chosen up front. Then,
\begin{align}
 |v(x_1, x_2;\tau) -1 |
 &\leq \frac{U^2-L^2}{L^2} + \frac{MU}{L^2}\omega^2\tau - C\tau + C\tau \\
 &= \frac{U^2-L^2}{L^2} - \frac{U^2-L^2}{L^2\tau_0}\tau + C\tau \\
 &= \left(1-\frac{\tau}{\tau_0}\right)\frac{U^2-L^2}{L^2} +C\tau.
\end{align}
The product on the RHS is non-positive for $\tau\geq\tau_0$, and hence $|v(x_1, x_2;\tau) -1 |\leq C\tau$ for all $\mathbb{R}^2\times [\tau_0, \infty)$.
\qed
\end{proof}

The bounds on $Q$ and $I$, and thus $v-1$, established in the preceding proof are not sufficient to conclude the result as $\tau\rightarrow 0$, unless $L=U$, which is the trivial case of a constant weighting function. More suitable bounds, however, can be found if inequality~\eqref{e;Taylor bounds on integral} can be inverted. This can be achieved by assuming $\tau$ to be small enough.

\begin{lemma}\label{t;v is of order tau for small tau}
Let $w$ be continuous and bounded away from zero, that is there exist $L$ and $U$ such that $0<L\leq w_{\boldsymbol{\theta}}(x) \leq U$ for all $x\in \mathbb{R}$. Let $w$ further be twice differentiable on $\mathbb{R}$ with $|w''(x)|<M$ for some $M$ and all $x\in \mathbb{R}$. Let $\tau_0 = \frac{2L}{M\omega^2}$. Then $v(x_1, x_2,;\tau) -1 = \mathcal{O}(\tau)$ on $\mathbb{R}^2\times (0,\tau_0)$.
\end{lemma}

\begin{proof}
Here we develop bounds on $Q$ and $I$ such that both $Q-1$ and $I-1$ are in the order of $\tau$. Let $\tau\leq\tau_{0}$ for $\tau_0$ as defined in the lemma. Then the lower bound of equation \eqref{e;Taylor bounds on integral} is bounded away from zero,
\begin{equation}
 w_{\boldsymbol{\theta}}(x) - \frac{M}{2}\omega^2\tau \geq w_{\boldsymbol{\theta}}(x) - \frac{M}{2}\omega^2\tau_0 > w_{\boldsymbol{\theta}}(x) - \frac{M}{2}\omega^2 \frac{2L}{M\omega^2} = w_{\boldsymbol{\theta}}(x) - L\geq 0.
\end{equation}
Hence we can invert the inequality~\eqref{e;Taylor bounds on integral} and obtain
\begin{equation}\label{e;bounds on Q}
 \frac{w_{\boldsymbol{\theta}}(x)-M\omega^2\tau}{w_{\boldsymbol{\theta}}(x)+\frac{M}{2}\omega^2\tau} \leq Q(x;\tau) \leq \frac{w_{\boldsymbol{\theta}}(x)+M\omega^2\tau}{w_{\boldsymbol{\theta}}(x)-\frac{M}{2}\omega^2\tau}.
\end{equation}
Note that the values in the numerators and denominators differ slightly because the variances of the kernel $k$ in the numerator and denominator of $Q$ differ by a factor of 2.

Since $2w_{\boldsymbol{\theta}}(x)-M\omega^2\tau \geq 2L-M\omega^2\tau_0>0$, we can conclude
\begin{equation}
 Q(x;\tau) - 1
 \leq \frac{w_{\boldsymbol{\theta}}(x)+M\omega^2\tau - w_{\boldsymbol{\theta}}(x)-\frac{M}{2}\omega^2\tau}{w_{\boldsymbol{\theta}}(x)-\frac{M}{2}\omega^2\tau}
 = \frac{M\omega^2\tau}{2w_{\boldsymbol{\theta}}(x)-M\omega^2\tau}
 \leq \frac{M\omega^2\tau}{2L-M\omega^2\tau_0},
\end{equation}
for all $x\in\mathbb{R}$ and $\tau<\tau_0$. Using $2w_{\boldsymbol{\theta}}(x)+M\omega^2\tau \geq 2w_{\boldsymbol{\theta}}(x)\geq 2L$, we similarly obtain,
\begin{equation}
 -(Q(x;\tau)-1) \leq \frac{3M\omega^2\tau}{2w_{\boldsymbol{\theta}}(x)+M\omega^2\tau} \leq \frac{3M}{2L}\omega^2\tau
\end{equation}
for all $x\in\mathbb{R}$ and $\tau<\tau_0$. If we set $C_1 := \max\left( \frac{M\omega^2}{2L-2\omega^2\tau_0},\frac{3M\omega^2}{2L} \right)$, it follows that \text{$\vert Q(x;\tau) - 1 \vert \leq C_1\tau$} on $\mathbb{R}^2\times (0,\tau_0)$.

Using analogous arguments as before, we can fine an find an upper bound on $I$,
\begin{align}
I(x_1, x_2;\tau) 
 &= \int_{\mathbb{R}} k_{\frac{\sigma}{\sqrt{2}}}\Bigl(z; \frac{1}{2}(x_1+x_2)\Bigr) \, \frac{w_{\boldsymbol{\theta}}(z)}{\int_{\mathbb{R}} k_{\sigma}(y;z) w_{\boldsymbol{\theta}}(y) \, dy} \, \D z \\ 
 &\leq \int_{\mathbb{R}} k_{\frac{\sigma}{\sqrt{2}}}\Bigl(z; \frac{1}{2}(x_1+x_2)\Bigr) \, \frac{w_{\boldsymbol{\theta}}(z)}{w_{\boldsymbol{\theta}}(z) - \frac{M}{2}\omega^2\tau} \, \D z \\
 &= \int_{\mathbb{R}} k_{\frac{\sigma}{\sqrt{2}}}\Bigl(z; \frac{1}{2}(x_1+x_2)\Bigr) \, \frac{w_{\boldsymbol{\theta}}(z)-\frac{M}{2}\omega^2\tau+\frac{M}{2}\omega^2\tau}{w_{\boldsymbol{\theta}}(z) - \frac{M}{2}\omega^2\tau} \, \D z \\
 &= \int_{\mathbb{R}} k_{\frac{\sigma}{\sqrt{2}}}\Bigl(z; \frac{1}{2}(x_1+x_2)\Bigr) \,\D z + \int_{\mathbb{R}} k_{\frac{\sigma}{\sqrt{2}}}\Bigl(z; \frac{1}{2}(x_1+x_2)\Bigr) \, \frac{\frac{M}{2}\omega^2\tau}{w_{\boldsymbol{\theta}}(z) - \frac{M}{2}\omega^2\tau} \, \D z \\
 &\leq 1 + \int_{\mathbb{R}} k_{\frac{\sigma}{\sqrt{2}}}\Bigl(z; \frac{1}{2}(x_1+x_2)\Bigr) \, \frac{\frac{M}{2}\omega^2\tau}{L - \frac{M}{2}\omega^2\tau_0} \, \D z 
 = 1+ \frac{M\omega^2\tau}{2L - M\omega^2\tau_0}.
\end{align}
A lower bound is given by
\begin{align}
 I(x_1, x_2;\tau)
 &\geq \int_{\mathbb{R}} k_{\frac{\sigma}{\sqrt{2}}}\Bigl(z; \frac{1}{2}(x_1+x_2)\Bigr) \, \frac{w_{\boldsymbol{\theta}}(z)}{w_{\boldsymbol{\theta}}(z) + \frac{M}{2}\omega^2\tau} \, \D z \\
 &= 1 - \int_{\mathbb{R}} k_{\frac{\sigma}{\sqrt{2}}}\Bigl(z; \frac{1}{2}(x_1+x_2)\Bigr) \, \frac{\frac{M}{2}\omega^2\tau}{w_{\boldsymbol{\theta}}(z) + \frac{M}{2}\omega^2\tau} \, \D z 
 \geq 1-\frac{M\omega^2\tau}{2L}.
\end{align}
Setting $C_2 := \frac{M\omega^2\tau}{2L - M\omega^2\tau_0}$, we obtain $|I(x_1,x_2;\tau)-1| \leq C_2\tau$ on $\mathbb{R}^2\times (0,\tau_0)$.

We can now estimate $v-1$ as follows,
\begin{align}
 |v(x_1,x_2;\tau) - 1| 
 &= |Q_{\tau}\, I_{\tau} -1| 
 \leq |Q_{\tau}-1|\,|I_{\tau}-1| + |Q_{\tau}-1| + |I_{\tau}-1| \\
 &\leq C_1\,C_2\tau^2+(C_1+C_2)\,\tau 
 \leq \big(C_1\,C_2\tau_0+C_1+C_2\big)\,\tau,
\end{align}
for all $x_1, x_2\in\mathbb{R}$ and all $\tau<\tau_0$.
\qed
\end{proof}

Lemmata~\ref{t;v is of order tau for large tau} and~\ref{t;v is of order tau for small tau}, together with proposition~\ref{t;p times v} prove Theorem~\ref{t;asymptotic robustness}.


\clearpage

\section{Supplemental Figures}
\label{s;supp figures}

\begin{figure}[h!]
\begin{center}
\includegraphics[width=0.9\textwidth]{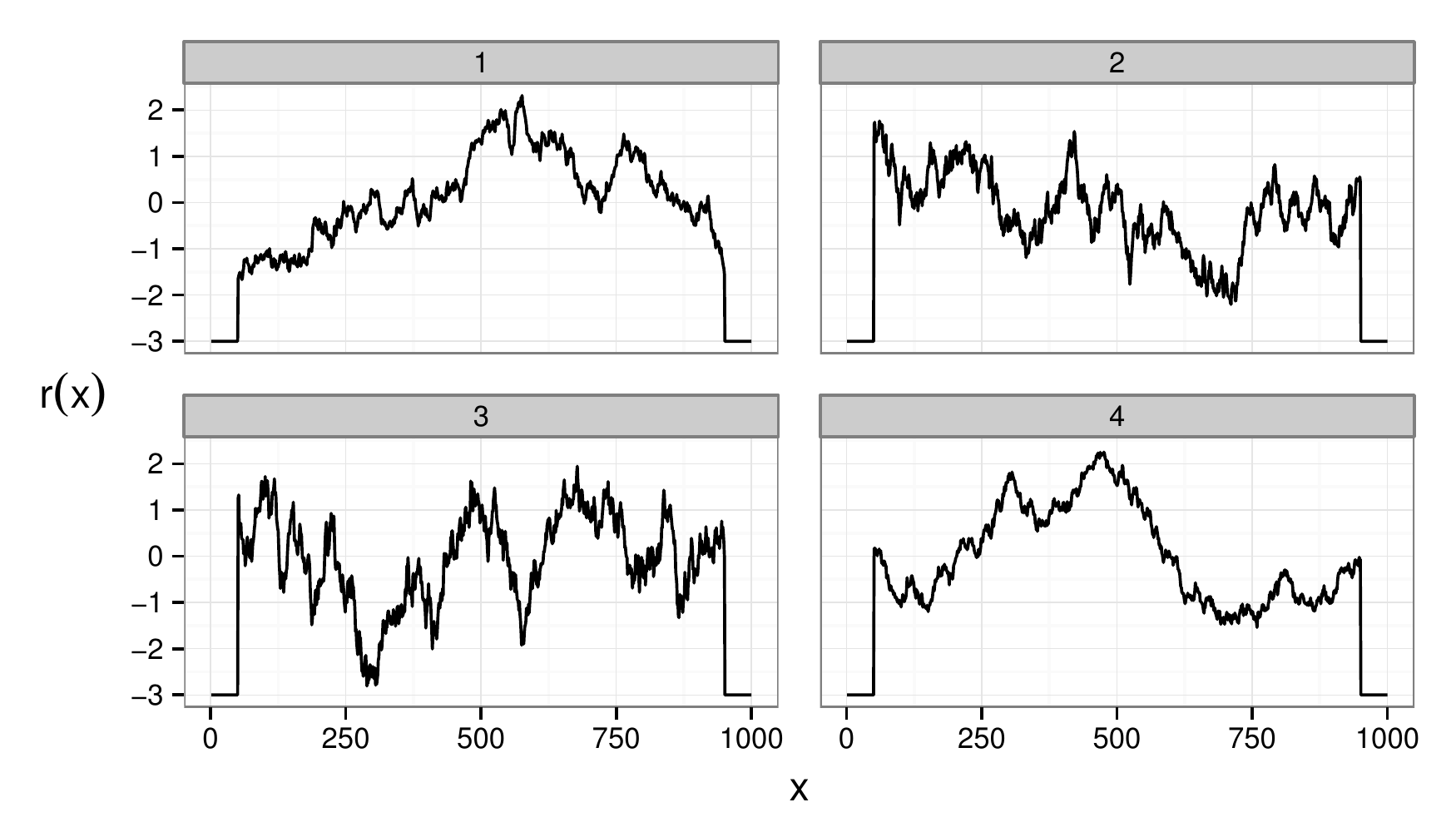}
\caption{Four of the simulated resource landscapes used for sampling movement trajectories. The depicted landscapes have been generated with spatial autocorrelation $\text{Cov}(r(x),r(y))=\exp\bigl(\frac{|x-y|}{s}\bigr)$ for $s=200,300,400,500$. We standardized landscapes to range within the interval $(-3,3)$. At the boundaries, we set values to -3 to avoid movement close to the boundary and thus boundary effects in the transition densities due to the normalization constant.}
\label{f;lands}
\end{center}
\end{figure}

\begin{figure}[h!]
\begin{center}
\includegraphics[width=0.8\textwidth]{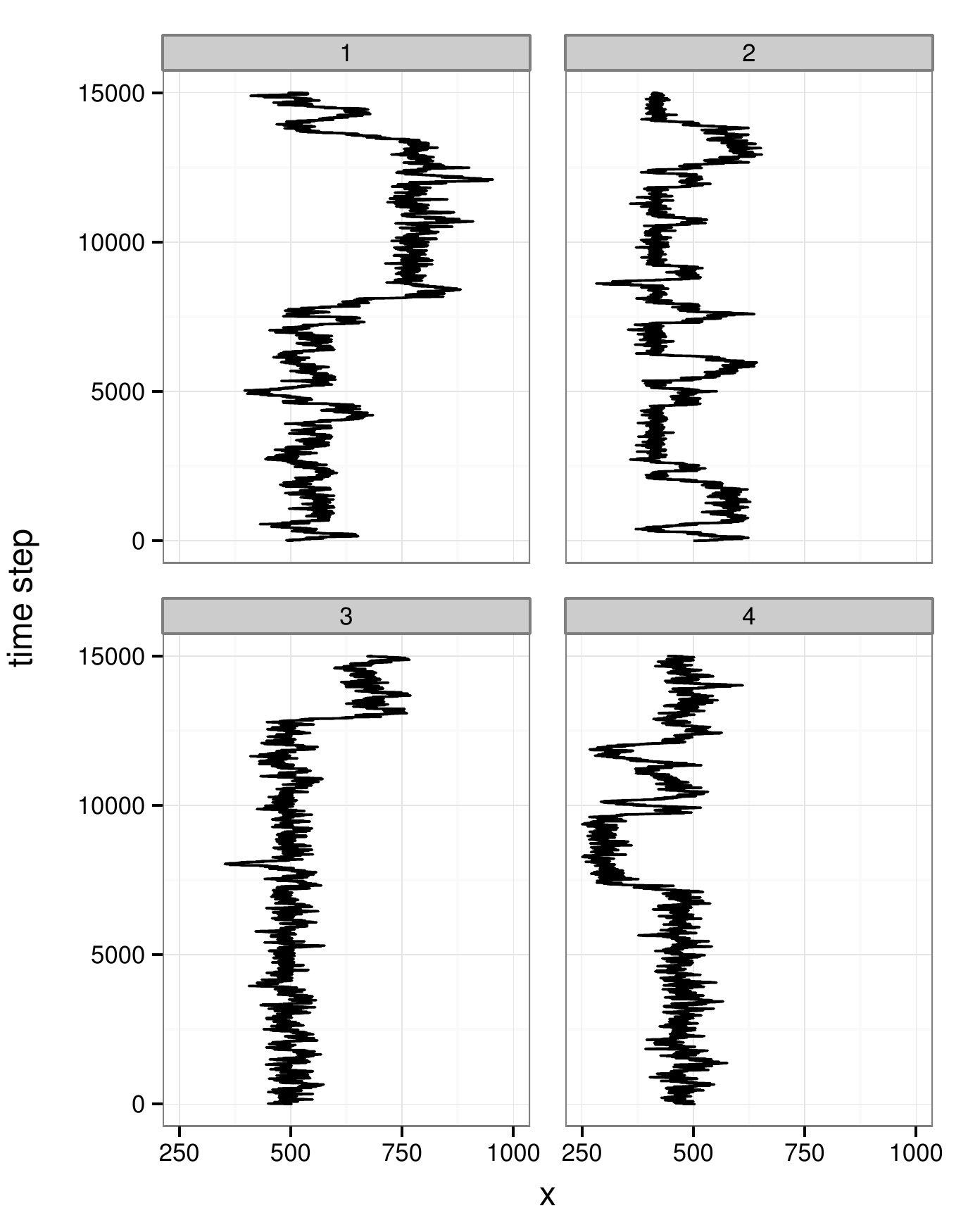}
\caption{Four of the simulated trajectories from the model with exponential RSF. The trajectories were generated using the parameter values $\sigma=6$ and $\beta=1$. The underlying resource landscapes are the landscapes depicted in Fig.~\ref{f;lands}, in the same order.}
\label{f;traj exp}
\end{center}
\end{figure}

\begin{figure}[h!]
\begin{center}
\includegraphics[width=0.8\textwidth]{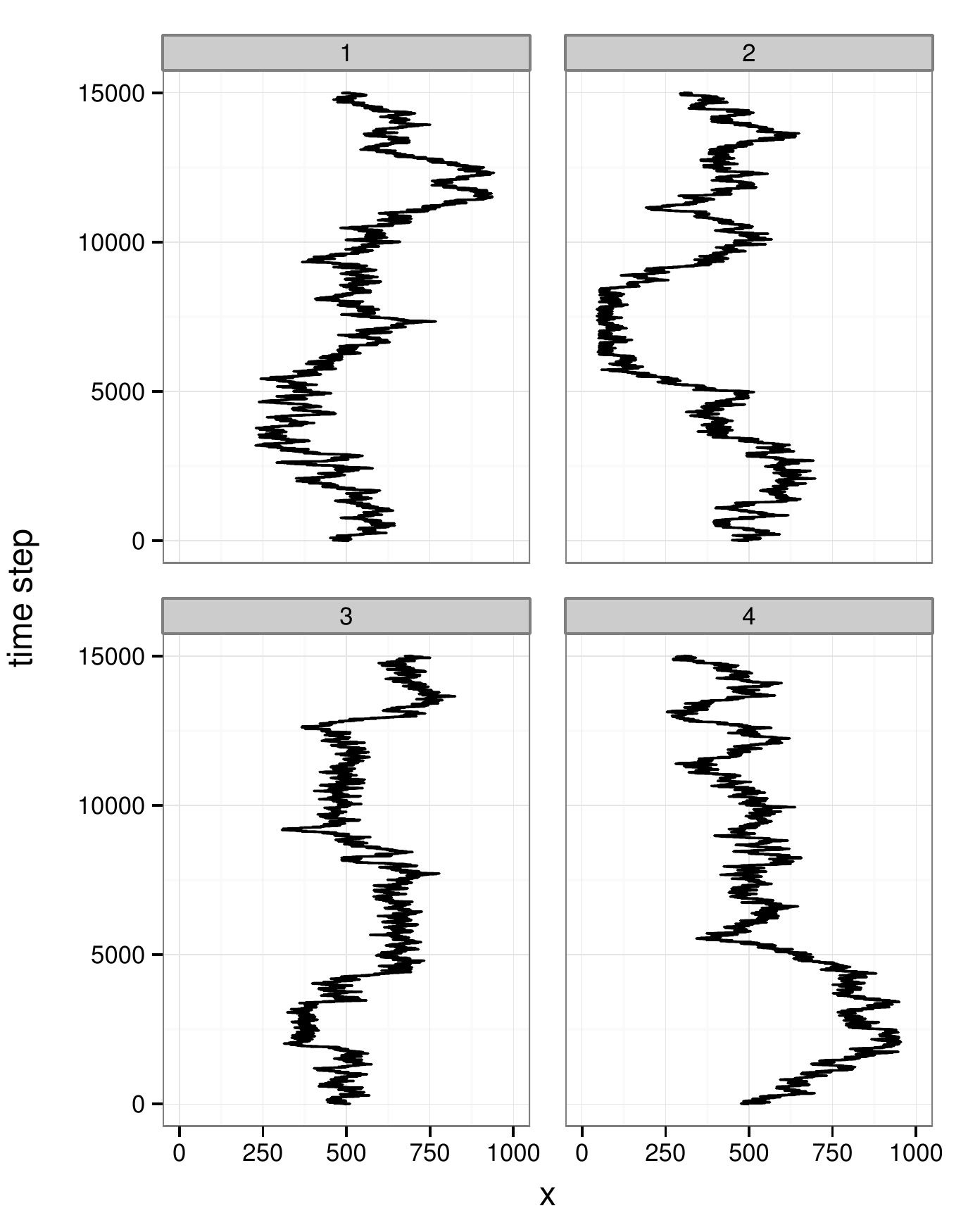}
\caption{Four of the simulated trajectories from the model with logistic RSF. The trajectories were generated using the parameter values $\sigma=6$ and $\beta=1$ (same as in Fig.~\ref{f;traj exp}) and $\alpha=0$. The underlying resource landscapes are again the landscapes depicted in Fig.~\ref{f;lands}, in the same order.}
\label{f;traj log}
\end{center}
\end{figure}

\end{document}